\documentclass[runningheads]{llncs}
\usepackage[T1]{fontenc}
\usepackage[utf8]{inputenc}

\usepackage{graphicx}

\usepackage{tikz}
\usetikzlibrary{arrows,calc,shapes}

\usepackage[bookmarksdepth=1,bookmarksopen=true]{hyperref}
\usepackage{color}

\usepackage{tcolorbox}
\newcommand{\llbox}[1]{
	\vspace{-1pt}
	\begin{tcolorbox}[width=\columnwidth, colframe=black, boxrule=0.25mm, top=1mm, left=1mm, right=1mm, bottom=1mm]
		#1
	\end{tcolorbox}
}

\usepackage{colortbl}
\usepackage{booktabs}
\usepackage{dcolumn}
\newcolumntype{d}[1]{D{.}{.}{#1} }
\newcolumntype{s}[1]{D{/}{/}{#1} }

\usepackage{amsmath}
\usepackage{amssymb}
\usepackage{algorithm}
\usepackage[compatible,noend]{algpseudocode}
\usepackage{subcaption}
\usepackage{listings}
\usepackage{tikz}
\usetikzlibrary{automata}

\makeatletter
\tikzset{circle split part fill/.style  args={#1,#2}{%
 alias=tmp@name, 
  postaction={%
    insert path={
     \pgfextra{%
     \pgfpointdiff{\pgfpointanchor{\pgf@node@name}{center}}%
                  {\pgfpointanchor{\pgf@node@name}{east}}%
     \pgfmathsetmacro\insiderad{\pgf@x}
      \fill[#1] (\pgf@node@name.base) ([xshift=-\pgflinewidth]\pgf@node@name.east) arc
                          (0:180:\insiderad-\pgflinewidth)--cycle;
      \fill[#2] (\pgf@node@name.base) ([xshift=\pgflinewidth]\pgf@node@name.west)  arc
                           (180:360:\insiderad-\pgflinewidth)--cycle;            
         }}}}}  
 \makeatother

\newcommand{\OUTPUT}{\ENSURE}

\newcommand{\waitlist}{\textit{waitlist}}
\newcommand{\visited}{\textit{N}_\wr} 
\newcommand{\bad}{\Delta}
\newcommand{\assign}{\mathbin{:=}}

\newcommand{\modified}{V_\delta}
\newcommand{\edges}{E_\wr}
\newcommand{\rd}{\textit{rd}}
\newcommand{\writ}{\textit{wrt}}
\newcommand{\postProcess}[2]{\textsc{process}(#1, #2)}

\RequirePackage[group-digits=integer, group-minimum-digits=4,
                list-final-separator={, and }, 
                free-standing-units, unit-optional-argument, 
                detect-weight=true, 
                per-mode=fraction, 
                ]{siunitx}

\usepackage{relsize}
\usepackage{xspace}

\newcommand{\reducer}{red}
\newcommand{\native}{nat}
\newcommand{\diffSyn}{syn}
\newcommand{\diffOurs}{DP}
\newcommand{\diffSynRed}{\ensuremath{\Delta^\mathrm{\reducer}_\mathrm{\diffSyn}}}
\newcommand{\diffSynNat}{\ensuremath{\Delta^\mathrm{\native}_\mathrm{\diffSyn}}}
\newcommand{\diffOursRed}{\ensuremath{\Delta^\mathrm{\reducer}_\mathrm{\diffOurs}}}
\newcommand{\diffOursNat}{\ensuremath{\Delta^\mathrm{\native}_\mathrm{\diffOurs}}}
\newcommand{\diffSynUn}{\ensuremath{\Delta_\mathrm{\diffSyn}}}
\newcommand{\diffOursUn}{\ensuremath{\Delta_\mathrm{\diffOurs}}}

\newcommand\tool[1]{{{\smaller\scshape #1}\xspace}}
\newcommand\definetool  [2]{\newcommand  {#1}{\tool{#2}\xspace}}
\definetool{\benchexec}   {BenchExec}

\definetool{\diffCond} {diff\-Cond}
\definetool{\diffDetectOurs}{diffDP} 
\definetool{\cpacheckerSVCOMP}  {CPAchecker~2.1}

\definetool{\cpachecker}  {CPAchecker}
\definetool{\cpacheckerDiffSyn}  {CPAchecker\(^{\diffSynRed}\)}
\definetool{\cpacheckerDiffOurs}  {CPAchecker\(^{\diffOursRed}\)}
\definetool{\esbmc}       {ESBMC}
\definetool{\esbmcDiffSyn}       {ESBMC\(^{\diffSynRed}\)}
\definetool{\esbmcDiffOurs}       {ESBMC\(^{\diffOursRed}\)}
\definetool{\predicate}  {Predicate}
\definetool{\predicateDiffSyn}  {Predicate\(^{\diffSynRed}\)}
\definetool{\predicateDiffSynNative}  {Predicate\(^{\diffSynNat}\)}
\definetool{\predicateDiffOurs}  {Predicate\(^{\diffOursRed}\)}
\definetool{\predicateDiffOursNative}  {Predicate\(^{\diffOursNat}\)}
\definetool{\utaipan}     {UTaipan}
\definetool{\utaipanDiffSyn}     {UTaipan\(^{\diffSynRed}\)}
\definetool{\utaipanDiffOurs}     {UTaipan\(^{\diffOursRed}\)}

\usepackage{pifont}
\usepackage{fontawesome}

\newcommand{\cproof}{\textcolor{green!80!black}{\ding{51}}}
\newcommand{\calarm}{\textcolor{red!80!black}{\ding{55}}}
\newcommand{\iresult}{\faBolt$_\textnormal{\cproof}$}

\makeatletter
\RequirePackage{hyperref}%

\def\orcidID#1{{\href{https://orcid.org/#1}{\protect\raisebox{-1.25pt}{\protect\includegraphics{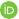}}}}}
\makeatother

\spnewtheorem*{theorem1}{Theorem~1}{\normalfont\bfseries}{\itshape}
\spnewtheorem*{theorem2}{Theorem~2}{\normalfont\bfseries}{\itshape}

\begin{document}
\title{Incorporating Data Dependencies and Properties in Difference Verification with Conditions (Technical Report)}

\titlerunning{Incorporating Data Dependencies and Properties in Difference Verification}

\author{
Marie-Christine Jakobs\inst{1,2}\,\orcidID{0000-0002-5890-4673}\and
Tim Pollandt\inst{1}\,\orcidID{0000-0001-9957-2444}
}
\authorrunning{M.-C.~Jakobs \and T.~Pollandt}

\institute{
Technical University of Darmstadt, Department of Computer Science
\and
LMU Munich
}

\maketitle              

\begin{abstract}
Software changes frequently.
To efficiently deal with such frequent changes, software verification tools must be incremental.
Most of today's approaches for incremental verification consider one specific verification approach. 
One exception is difference verification with conditions recently proposed by Beyer et al.
Its underlying idea is to determine an overapproximation of those modified execution paths that may cause a new property violation, which does not exist in the unchanged program, encode the determined paths into a condition, and use the condition to restrict the verification to the analysis of those determined paths.
To determine the overapproximation, Beyer et al.\ propose a syntax-based difference detector that adds any syntactical path of the modified program that does not exist in the original program into the overapproximation.
This paper provides a second difference detector~\diffDetectOurs, which computes a more precise overapproximation by taking data dependencies and program properties into account when determining the overapproximation of those modified execution paths that may cause a new property violation.
Our evaluation indeed shows that our more precise difference detector improves the effectiveness and efficiency of difference verification with condition on several tasks.

\keywords{Incremental verification \and Regression verification \and Difference verification \and Conditional model checking \and Software verification.}
\end{abstract}

\section{Introduction}
Software is omnipresent and in many cases we expect it to work reliably. 
Meanwhile, bug fixes, change requests, or the addition of new features, etc.\ cause frequent software changes and each of these changes may introduce new bugs and, thus, may lower the software's reliability.
One countermeasure for this is to use software verification~\cite{DBLP:journals/tcad/DSilvaKW08,DBLP:journals/csur/JhalaM09} to prove the absence or detect the existence of software bugs.
Performing the verification on the complete modified program after each change is typically too costly, but often also not necessary because a modification rarely has the potential to introduce new bugs in all parts of the software. 
This observation is leveraged by many incremental software verification techniques, which aim to speed up the reverification of modified software.
For example, there exist techniques that update previously computed state space descriptions~\cite{Reviser,IncA,DBLP:conf/scam/PlasSER20,ExtremeMC,ISSE,EvolCheck}, reuse intermediate results~\cite{greenModelChecking,Recal,PrecisionReuse,TraceAbstractionReuse,SummaryReuse,AnnotationReuse,P2V}, or skip the analysis of unchanged behavior~\cite{RegressionMC,DiSE,iDiSE,Memoise,LeinoW15,iCoq,RPS,DiffCond}.
However, most approaches are tailored to and coupled with one specific verification approach. 

\begin{figure}[t]
	\centering
		\begin{tikzpicture}[>=stealth]
		  
			\node (de) [fill=black!10!white, draw, rounded corners, minimum width=17.5em, minimum height=5em] {};
			\node at ($(de.south) +(-0.6,0)$) [anchor=south] {Difference Condition Extractor};
			\node (dc) at ($(de)+(-0.65,0.22)$) [anchor=east, fill=lime!80!black, text width=5em, align=center] {Difference Detector};
			\node (cg) at ($(de)+(0.65,0.22)$) [anchor=west, fill=teal!50!white, text width=5em, align=center] {Condition Generator};
			\node (mp) [above of = dc, anchor=south, node distance=0.75cm] {modified program~\(P'\)};
			\node (op) [below of = dc, anchor=north, node distance=1.25cm] {original program~\(P\)};

			\node (cv) [right of =cg, anchor=west, node distance=3cm, draw, rounded corners, fill=black!10] {Conditional Verifier};
			\node (rcv) [below of = cv, node distance=1.25cm, minimum width=5.7cm, minimum height=1.35cm, draw, rounded corners, fill=black!10] {};
			\node at (rcv.south) [anchor=south] {Reducer-based Conditional Verifier~\cite{ReducerCMC}};
			\node (r) at ($(rcv)+(-0.75,0.2)$) [anchor=east, fill=cyan] {Reducer};
			\node (v) at ($(rcv)+(0.75,0.2)$) [anchor=west, fill=orange] {Verifier};
			
			\draw[->] (op.north) -| (dc);
			\draw[->] (mp.south) -| (dc);
			\draw[->] (mp) -| (cv);
			\draw[->] (dc) to node[above, font=\scriptsize] {DG(P,P')} (cg);
			\draw[->] (cg) to node[above, font=\scriptsize] {~condition} (cv);
			\draw[->] (r) to node[above, font=\scriptsize]{residual} (v);
			\draw[->] (r) to node[below, font=\scriptsize] {program} (v);
			\draw[->,>=open triangle 45] (rcv) -- (cv);
		\end{tikzpicture}
	\caption{Construction of a difference verifier}
	\label{fig:differenceVerifier}
\end{figure}
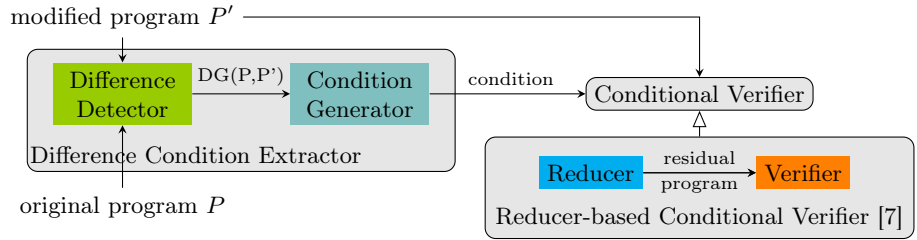

Recently, Beyer et al.\ suggested difference verification with conditions~\cite{DiffCond}, an incremental verification approach agnostic to the verifier.
Difference verification with conditions allows one to build a difference verifier (see Fig.~\ref{fig:differenceVerifier}) from an arbitrary (conditional) verifier\footnote{We can use the reducer-based approach~\cite{ReducerCMC} to conditional model checking to turn an arbitrary verifier into a condiitonal verifier.} and let the difference verifier explore that part of the state space that may contain new bugs introduced by the software modification.
The respective difference verifier first runs a difference condition extractor.
The difference condition extractor detects an overapproximation of all execution paths of the modified program that may cause a regression bug, i.e., execution paths that may violate the property and that are triggered by inputs that may not cause a property violation in the original program, and then encodes the detected paths into a condition, an automata-based exchange format used in conditional model checking~\cite{CMC}.
In a second step, the difference verifier feeds the condition and the modified program into its conditional verifier to restrict the verification to those paths of the modified program that were determined and encoded into the condition by the extractor in the previous step. 

In their proof-of-concept, Beyer et al.\ used a syntax-based difference condition extractor~\cite{DiffCond} that determines exactly those syntactic paths of the modified program that syntactically changed, i.e., all syntactic paths that do not occur in the original program.
The overapproximation computed by their syntax-based difference condition extractor is rather coarse.
For instance, consider the modification of the program in Fig.~\ref{fig:exampleprograms}.
The modification deletes the minus sign in the first statement (shown in red) to fix a bug in the computation of the negated absolute value of \(x\).
Despite the modification, due to a missing data dependency between the change and the property (i.e., the assertion), the property is not affected.
To take this into account, this paper introduces a new, more complex difference condition extractor. 
The paper proves the soundness of the new difference condition extractor and experimentally applies it in various difference verifiers, which are compared to verification without the difference extractor and difference verification with the syntax-based extractor.

This technical report is an extension of our conference paper~\cite{diffDP} and enhances our conference paper with the soundness proofs. To be self-contained, the technical report presents all contributions of the conference paper~\cite{diffDP}.

\section{Foundations of Difference Verification with Conditions} 
We present difference verification with conditions in the context of an imperative programming language that operates on assume operations and assignments.
We assume that a set~\(Ops\) contains all these operations.
Furthermore, we restrict our properties to the unreachability of certain program locations.
Note that our implementation supports C~programs and many safety properties can be encoded by the unreachability of an error location.
Formally, we represent programs including their properties by \emph{control-flow automata}.
\begin{definition}
A \emph{control-flow automaton} (CFA) \(P=(L,\ell_0,G, \ell_\mathrm{err})\) consists of
\begin{itemize}
	\item a set \(L\) of program locations with initial location \(\ell_0\in L\) and error location \(\ell_\mathrm{err}\in L\), which when reached reflects a property violation, and
	\item a set \(G\subseteq L\times Ops\times L\) of control-flow edges.
\end{itemize}
\emph{\(P\) is deterministic} if for all control-flow edges \((\ell, op_1, \ell_1), (\ell, op_2, \ell_2)\in G\) either \(op_1=op_2\) and \(\ell_1=\ell_2\), or \(op_1\) and \(op_2\) are assume operations with \(op_1\equiv\neg op_2\).
\end{definition}
\begin{figure}[t]
	\begin{subfigure}[b]{0.24\textwidth}
		\centering
		\begin{lstlisting}[basicstyle=\ttfamily,language=C,tabsize=3,emphstyle={\color{red}},escapechar=\%]
r=%\color{red}{\underline{-}}%x;
if (x>0) {
	r=-x;
	assert(r<=0);
}
	\end{lstlisting}
\end{subfigure}
\hfill
\begin{subfigure}[b]{0.3\textwidth}
	\centering
	\begin{tikzpicture}
		\footnotesize
		\node[draw, circle, inner sep=0cm,fill=yellow!20] (l0) at (0,1.5) {$\ell_0$};
		\node[draw, circle, inner sep=0cm,fill=yellow!20] (l1) at (0,0.75) {$\ell_1$};
		\node[draw, circle, inner sep=0cm,fill=yellow!20] (l2) at (1.25,0.75) {$\ell_2$};
		\node[draw, circle, inner sep=0cm,fill=yellow!20] (l3) at (1.25,0) {$\ell_3$};
		\node[draw, circle, inner sep=0cm,fill=yellow!20] (l4) at (0,0) {$\ell_4$};
		\node[draw, circle, inner sep=0cm,thick,accepting,fill=yellow!20] (l5) at (2.5,0) {$\ell_\mathrm{err}$};
		\path[->,draw] (-0,1.8) -- (l0);
		\path[->,draw,red] (l0) -- node[left,font=\footnotesize]{\texttt{r=-x;}} (l1);
		\path[->,draw] (l1) -- node[above,font=\footnotesize]{{\texttt{x>0}}} (l2);
		\path[->,draw] (l1) -- node[left,font=\footnotesize]{{\texttt{!(x>0)}}} (l4);
		\path[->,draw] (l2) -- node[right,font=\footnotesize]{\texttt{r=-x;}} (l3);
		\path[->,draw] (l3) -- node[sloped,below, yshift=-2ex, xshift=-0.5ex,font=\footnotesize]{\texttt{r<=0;}} (l4);
		\path[->,draw] (l3) -- node[sloped, below, yshift=-2ex, xshift=0.5ex,font=\footnotesize]{\texttt{!(r<=0)}} (l5);
	\end{tikzpicture}
\end{subfigure}
\hfill
\begin{subfigure}[b]{0.34\textwidth}
	\centering
	\begin{tikzpicture}
		\footnotesize
		\node[draw, circle, inner sep=0cm,fill=yellow!20] (l0) at (0,1.5) {$\ell'_0$};
		\node[draw, circle, inner sep=0cm,fill=yellow!20] (l1) at (0,0.75) {$\ell'_1$};
		\node[draw, circle, inner sep=0cm,fill=yellow!20] (l2) at (1.25,0.75) {$\ell'_2$};
		\node[draw, circle, inner sep=0cm,fill=yellow!20] (l3) at (1.25,0) {$\ell'_3$};
		\node[draw, circle, inner sep=0cm,fill=yellow!20] (l4) at (0,0) {$\ell'_4$};
		\node[draw, circle, inner sep=0cm,thick,accepting,fill=yellow!20] (l5) at (2.5,0) {$\ell'_\mathrm{err}$};
		\path[->,draw] (-0,1.8) -- (l0);
		\path[->,draw,blue] (l0) -- node[left,font=\footnotesize]{\texttt{r=x;}} (l1);
		\path[->,draw] (l1) -- node[above,font=\footnotesize]{{\texttt{x>0}}} (l2);
		\path[->,draw] (l1) -- node[left,font=\footnotesize]{{\texttt{!(x>0)}}} (l4);
		\path[->,draw] (l2) -- node[right,font=\footnotesize]{\texttt{r=-x;}} (l3);
		\path[->,draw] (l3) -- node[sloped,below, yshift=-2ex, xshift=-0.5ex,font=\footnotesize]{\texttt{r<=0;}} (l4);
		\path[->,draw] (l3) -- node[sloped, below, yshift=-2ex, xshift=0.5ex,font=\footnotesize]{\texttt{!(r<=0)}} (l5);
	\end{tikzpicture}	
\end{subfigure}
\vspace{-1em}
\caption{Program code of original program (left) as well as the CFA of original program (middle) and modified program (right). The modified program is derived from the original program by deleting the underlined parts in red.}
\label{fig:exampleprograms}
\vspace{-1em}
\end{figure}
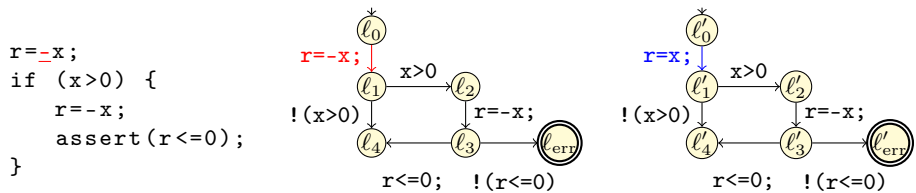
For the explanation throughout the paper, we use the original and modified program shown Figure~\ref{fig:exampleprograms}.
On the left of the figure, we show the code of our original program, which should compute the negated absolute value $-|x|$ of variable~\(x\). 
However, the red, underlined minus sign in the first statement causes a bug and is removed in the modified program, whose CFA is shown on the right of Fig.~\ref{fig:exampleprograms}.
The middle of Fig.~\ref{fig:exampleprograms} presents the CFA of the original, which only differs from the CFA of modified program in the first edge (highlighted in red and blue). 
Both CFAs use one edge per assignment and two edges for the condition in the if statement as well the assertion, namely one per condition evaluation.
In addition any violation of the assertion leads to the error location represented by double circles.

CFAs define the program syntax. 
The semantics of a CFA~\(P=(L,\ell_0,G, \ell_\mathrm{err})\) is defined by a standard operational semantics. 
Therefore, we represent a program state as as pair of a program location~\(\ell\in L\) and a concrete data state~\(c\) that assigns to each variable a concrete value.
Based on the notion of program states, we define the executable program paths.
An \emph{executable program path}~$\pi = (\ell_0,c_0) \stackrel{g_1}{\rightarrow}(\ell_1, c_1) \stackrel{g_2}{\rightarrow} \ldots \stackrel{g_n}{\rightarrow} (\ell_n, c_n)$ always begins in the initial location~\(\ell_0\) and an arbitrary concrete data state.
Thereafter, it executes \(n\) transition steps with \(n\geq0\).
Thereby, it ensures that each of the \(n\) transition 
steps~\(i\), i.e., \(1\leq i\leq n\), respects  the control-flow, i.e., 
\(g_i=(\ell_{i-1},\cdot,\ell_i)\), and that it adheres to the operation semantics, i.e., (a) for of assume operations, \(c_{i-1}\models op_i\) and \(c_{i-1}=c_i\) and (b) for assignments, \(c_i=SP_{op_i}(c_{i-1})\), where \(SP\) denotes the strongest-post operator of the semantics.
Next to semantics, we define the set~\(\writ(op)\) that contains the variables written by an operation, i.e., the variables whose value may be changed by operation~\(op\).
Similarly, we define the set~\(\rd(op)\) that contains all variables necessary to execute the operation~\(op\), i.e., if two data states \(c, c'\) agree on the values of the variables in \(\rd(op)\), then for assume operations \(c\models op\Leftrightarrow c'\models op\) and for assignments, \(SP_{op}(c)\) and \(SP_{op}(c')\) agree on the values of the variables in \(\writ(op)\).

Based on the notion of execution paths, we now charachterize those paths violate the program's property.
A path violates the program's property whenever it reaches the error location, i.e., \(\exists 0\leq i\leq n: \ell_i=\ell_\mathrm{err}\).
We write \(paths^\mathrm{err}(P)\) to denote the set of all execution paths of program~\(P\) that violate a property and use it to characterize all execution paths of a modified program~\(P'=(L',\ell'_0,G', \ell'_\mathrm{err})\) that cause a regression bug.
More concretely, these are all paths from \(paths^\mathrm{err}(P')\) whose initial concrete data state~\(c'_0\) does not trigger a property violation in the the original program \(P=(L,\ell_0,G, \ell_\mathrm{err})\).
Formally, a path $\pi' = (\ell'_0,c'_0) \stackrel{g'_1}{\rightarrow}(\ell'_1, c'_1) \stackrel{g'_2}{\rightarrow} \ldots \stackrel{g'_n}{\rightarrow} (\ell'_n, c'_n)\in paths^\mathrm{err}(P')$ causes a regression bug if there does not exist a path $\pi = (\ell_0,c_0) \stackrel{g_1}{\rightarrow}(\ell_1, c_1) \stackrel{g_2}{\rightarrow} \ldots \stackrel{g_m}{\rightarrow} (\ell_m, c_m)\in paths^\mathrm{err}(P)$ with \(c'_0=c_0\).
We use the set~\(paths^\mathrm{rb}(P, P')\) to describe all \emph{paths with regression bugs}.

In most cases, one fails to precisely compute the paths with regression bugs. 
Thus, one needs to compute an overapproximation of it and then analyze the paths in the overapproximation.
To analyze the paths in the overapproximation, difference verification with conditions applies the idea of conditional model checking~(CMC)~\cite{CMC,ReducerCMC}.
CMC steers a conditional verifier (see right of Fig.~\ref{fig:differenceVerifier}) with a condition to restrict the verification to the subset of the program paths that are not covered by the condition.
Following CMC~\cite{CMC,ReducerCMC}, we represent a condition by an automaton and use accepting states to identify covered paths.
\begin{definition}
A \emph{condition} $A=(Q,\delta,q_0,F)$ consists of
\begin{itemize}
	\item a finite set $Q$ of states including the initial state $q_0\in Q$ and a set $F\subseteq Q$ of accepting states and
	\item a transition relation $\delta\subseteq Q\times G\times Q$\footnote{In general~\cite{CMC,ReducerCMC},
      the transition relation of a condition also specifies assumptions on the program states, which we omit because they are irrelevant in difference verification.} 
    ensuring $\forall (q, op, q')\in\delta: q\!\notin\! F$.
\end{itemize}
\end{definition}
As already mentioned, CMC analyzes at least all paths that are not covered by the condition. 
Formally, a condition covers a path if the condition accepts any prefix of the path.
\begin{definition}
  A condition $A = (Q,\delta,q_0,F)$ {\em covers} an executable program path 
  $\pi = (\ell_0,c_0) \stackrel{g_1}{\rightarrow}(\ell_1, c_1) \stackrel{g_2}{\rightarrow} \ldots \stackrel{g_n}{\rightarrow} (\ell_n, c_n)$, i.e., \(\pi\in cover(A)\),
   if there exists a run $\rho = q_0 \stackrel{g_1}{\rightarrow} q_1 \stackrel{g_2}{\rightarrow} \ldots\stackrel{g_k}{\rightarrow} q_k, 0 \leq k \leq n$, in $A$, s.t.\
 $q_k \in F$.
\end{definition}
In difference verification with conditions, the difference verifier (see Fig.~\ref{fig:differenceVerifier}) first uses a difference condition extractor that generates an appropriate condition, which must not accept any path from \(paths^\mathrm{rb}(P, P')\), and then applies CMC.
We consider difference condition extractors that follow the procedure of the existing syntax-based extractor~\cite{DiffCond} and, thus, generate the condition in two steps.

In the first step, a \emph{difference detector} analyzes the modified program together with the original program to determine an overapproximation of  \(paths^\mathrm{rb}(P, P')\). 
The resulting overapproximation is described by a \emph{difference graph}~\(DG(P,P')\).
Paths in the difference graph refer to program paths of the modified program.
In particular, paths in the difference graph that reach nodes from a set~\(\Delta\) of regression bug indicator nodes characterize program paths that should be analyzed because they may cause regression bugs (like paths in \(paths^\mathrm{rb}(P, P')\)). 
To ensure this relation and to avoid that paths in the difference graph are  misclassified after reaching a node from \(\Delta\), a difference graph~\(DG(P,P')\) 
guarantees that all of its paths~\(p\) that relate to (a prefix of) a path with a regression bug must be extendable to a path \(p'\) in the difference graph that ends in a node from~\(\Delta\). 
Formally, we model difference graphs~\(DG(P,P')\) as follows.
\begin{definition}
A \(DG(P,P')=(N,E,n_0,\Delta)\) for original CFA~\(P=(L,\ell_0,G, \ell_\mathrm{err})\) and modified CFA~\(P'=(L',\ell'_0,G', \ell'_\mathrm{err})\) consists of 
\begin{itemize}
	\item
	a set~\(N\) of nodes including the initial node~\(n_0\in N\) as well as the set \(\Delta\subseteq N\) of regression bug indicator nodes, 
	\item 
	and edges~\(E\subseteq N\times G'\times N\). It must ensure the soundness property
 \end{itemize}
\begin{center}
 \(\begin{array}{c}\forall \pi' = (\ell'_0,c'_0) \stackrel{g'_1}{\rightarrow}(\ell'_1, c'_1) \stackrel{g'_2}{\rightarrow} \ldots \stackrel{g'_n}{\rightarrow} (\ell'_n, c'_n)\in paths^\mathrm{rb}(P,P'), 0{\leq}k{\leq}n, \\ n_0\stackrel{g'_1}{\rightarrow}n_1 \stackrel{g'_2}{\rightarrow} \ldots \stackrel{g'_k}{\rightarrow} n_k\in DG(P,P'): \exists n_k\stackrel{\cdot}{\rightarrow} \ldots \stackrel{\cdot}{\rightarrow} n_{k+m}: n_{k+m}\in\Delta\enspace.\end{array}\)
\end{center}
\end{definition}

\begin{algorithm}[t]
\caption{Condition generator extracting condition from difference graph}\label{alg:cond}
\begin{algorithmic}[1]
\REQUIRE  difference graph $\texttt{DG(P,P')} = (N, E, n_0, \Delta)$ 
\ENSURE extracted condition
\STATE $Q=\{n_0\}\cup \Delta$;\hspace{1em}$\mathtt{waitlist}=\{n_0\}\cup\Delta$; 
\WHILE{($\texttt{waitlist} \neq\emptyset$)}
  \STATE pop $n_s$ from \texttt{waitlist}
  \FOR{each $(n_p,(\ell,op,\ell'), n_s)\in E$ with $n_p\notin Q$}
		  \STATE $Q = Q \cup \{n_p\}$; ~ $\texttt{waitlist} = \texttt{waitlist} \cup \{n_p\}$;
	\ENDFOR
\ENDWHILE
\STATE F =  $\{n_s \mid \exists (n_p,g,n_s)\in E \wedge n_p\in Q\wedge n_s\notin Q\}$;
\STATE \textbf{return} ($Q\cup F, E \cap (Q\times G'\times (Q\cup F)), n_0, F$);
\end{algorithmic}
\end{algorithm}

In the second step of the difference condition extractor, the \emph{condition generator} converts the difference graph resulting from the first step into a condition.
Algorithm~\ref{alg:cond} performs that conversion.
To this end, it generates a condition that accepts a path as soon as a prefix has been seen for which it knows that any extension of a path in the difference graph that is related to the prefix cannot reach a node from \(\Delta\), i.e., due to the assumption on a difference graph the prefix cannot be extended to a path from \(paths^\mathrm{rb}(P, P')\).
Due to the requirements on the difference graph, the relevant paths in the difference graph start in \(n_0\), end in a node that cannot reach a node from \(\Delta\), and all other nodes on the paths may reach a node from \(\Delta\).  
Using these insights, Alg.~\ref{alg:cond}  performs a backward search from \(\Delta\) to the initial node~\(n_0\) (lines~1--5) to detect the non-accepting condition states (intermediate nodes of the paths, which may reach a  node from \(\Delta\)).
Then, it adds the accepting states (line~6), the end nodes of these paths, and copies all edges of the difference graph induced by the accepting and non-accepting states.

Next we show that difference condition extractors based on the above process and requirements, generate appropriate conditions and, thus, ensure that difference verifiers using them analyze all paths with regression bugs.
Therefore, we need to show that the generated condition~\(A\) does not accept paths causing regression bugs, i.e., \(cover(A)\cap paths^\mathrm{rb}(P, P')=\emptyset\).
\begin{theorem}\label{theo:extract}
Let \(P=(L,\ell_0,G, \ell_\mathrm{err})\) and \(P'=(L',\ell'_0,G', \ell'_\mathrm{err})\) be two CFAs. 
Difference condition extractors that sequentially compose a difference detector with Alg.~\ref{alg:cond} compute conditions~\(A\) that do not cover paths that cause regression bugs, i.e., \(cover(A)\cap paths^\mathrm{rb}(P, P')=\emptyset\).
\end{theorem}
\begin{proof}
Let \(P=(L,\ell_0,G, \ell_\mathrm{err})\) and \(P'=(L',\ell'_0,G', \ell'_\mathrm{err})\) be arbitrary CFAs,  \(DG(P,P')\) a difference graph computed by an arbitrary difference detector, and \(A=(Q, \delta, q_0, F)\) the condition computed by Alg.~\ref{alg:cond} for input \(DG(P,P')\).
Consider arbitrary path \(\pi' = (\ell'_0,c'_0) \stackrel{g'_1}{\rightarrow}(\ell'_1, c'_1) \stackrel{g'_2}{\rightarrow} \ldots \stackrel{g'_n}{\rightarrow} (\ell'_n, c'_n)\in cover(A)\).
By definition of coverage, there exists
 $\rho = q_0 \stackrel{g'_1}{\rightarrow} q_1 \stackrel{g'_2}{\rightarrow} \ldots\stackrel{g'_k}{\rightarrow} q_k, 0 \leq k \leq n$, in $A$, such that
 $q_k \in F$.
Due to the construction of \(A\) in Alg.~\ref{alg:cond}, there exists a path \(n_0 \stackrel{g'_1}{\rightarrow} n_1 \stackrel{g'_2}{\rightarrow} \ldots\stackrel{g'_k}{\rightarrow} n_k\in DG(P,P')\) such that  \(\neg\exists n_k\stackrel{\cdot}{\rightarrow} \ldots \stackrel{\cdot}{\rightarrow} n_{k+m}: n_{k+m}\in\Delta\) (otherwise \(n_k\) would have been added to \(Q\) in lines 1--5 and, thus, not to \(F\) in line~6).
Due to the soundness property of the difference graph, \(\pi'\notin paths^\mathrm{rb}(P, P')\).\qed
\end{proof}
Beyer et al.~\cite{DiffCond} use the two step process from above to define a syntax-based difference condition extractor~\diffCond. 
Their difference detector explores original and modified program in synch and stops exploration in a node from \(\Delta\) as soon as a syntactical difference between original and modified program occurs on the path. 
In the following section, we describe a second difference detector, which goes beyond syntactical differences and takes data dependencies and error locations into account.

\section{Difference Detector \diffDetectOurs} 
The crucial point to make difference verification with conditions efficient is to let the verifier explore a (small) subset of the program paths.
Since for soundness the verifier needs to explore at least all paths with regression bugs~\(paths^\mathrm{rb}(P, P')\), it is crucial that the difference detector, which is responsible for determining the subset of programs to be explored, computes a (small) overapproximation of the paths with regression bugs.
In this section, we present a new difference detector~\diffDetectOurs that 
takes data dependencies and program properties to compute a more precise overapproximation of \(paths^\mathrm{rb}(P, P')\) than the existing, syntax-based difference detector~\cite{DiffCond}. 

While the set of all execution paths of the modified program is a proper overapproximation of \(paths^\mathrm{rb}(P, P')\), difference detectors aim to be more precise and remove program paths that 
do not violate a property or for which they can determine that the violation already exists in the original program, i.e., it is not a path with a regression bug. 
To determine that a path~\(\pi'\) violating the property is not a path with a regression bug, one needs to ensure the existence of what we call a witness path, namely a path of the original program that starts with the same input as \(\pi'\) and also violates a property.
The existing syntax-based difference detector~\cite{DiffCond} only considers witness paths that are identical to \(\pi'\) except for location naming.
The difference detector~\diffDetectOurs presented in this section takes additional witness paths into account.
In general, it considers witness paths that have a similar branching structure as \(\pi'\), i.e., execute a similar sequence of assume operations, and violate properties no later than \(\pi'\).\footnote{All identical witness paths also fulfill the criterion of a similar branching structure.}

Since the number of execution paths is at least large and may be even infinite, we cannot detect witnesses per execution path.
Instead \diffDetectOurs groups execution paths following the same syntactic path and considers all paths of a group at once.
More concretely, for each syntactical path~\(p'\) of the modified program \diffDetectOurs tries to detect a syntactic (witness) path~\(p\) in the original program that (1)~ has a similar branching structure, that (2)~is executable whenever \(p'\) is executable (thus, ensuring executability of path~\(p\) for all inputs of interest), and (3)~in which an error location does not occur later than in \(p\). 
For example when considering the syntactic path \(p':=\ell'_0\xrightarrow{r=x;}\ell'_1\xrightarrow{x>0}\ell'_2\xrightarrow{r=-x;}\ell'_3\xrightarrow{!(r<=0)}\ell'_\mathrm{err}\) of our modified program in Fig.~\ref{fig:exampleprograms}, difference detector \diffDetectOurs detects the syntactic witness path \(p:=\ell_0\xrightarrow{r=-x;}\ell_1\xrightarrow{x>0}\ell_2\xrightarrow{r=-x;}\ell_3\xrightarrow{!(r<=0)}\ell_\mathrm{err}\). 
Note that \(p\) is executable whenever \(p'\) is executable because the changes (first assignment) do not affect the assume statements on the paths.

\begin{algorithm}[t]
  \caption{Difference detector \diffDetectOurs$(P, P')$}\label{alg:diffcompprop}
\begin{algorithmic}[1]
\REQUIRE deterministic CFAs $P=(L,\ell_0,G,\ell_{err})$, $P'=(L',\ell'_0,G',\ell'_{err})$ 
\OUTPUT difference graph $DG(P,P')$ 

\STATE $\waitlist \assign\{(\ell_0, \ell'_0)\}; \visited \assign \{(\ell_0, \ell'_0)\}; \edges \assign \emptyset;\bad\assign \emptyset; \modified\assign \lambda((\ell_1,\ell_1')).\emptyset;$ 
\WHILE{$\waitlist \neq \emptyset$}
  \STATE pop $(\ell_1,\ell'_1)$ from \waitlist
  \FOR{\textbf{each} $g'=(\ell'_1,op',\ell'_2) \in G'$}
    \IF{$op'$ is assignment}
			\IF{$\exists (\ell_1,op',\ell_2) \in G$}
				\IF{$\modified((\ell_1,\ell'_1))\cap \rd(op')\neq\emptyset$}
				   ~$V_\mathrm{diff}\assign\modified((\ell_1,\ell'_1))\cup\writ(op');$
			  \ELSE
				  ~$V_\mathrm{diff}\assign\modified((\ell_1,\ell'_1))\setminus\writ(op');$
				\ENDIF
				\STATE \postProcess{$((\ell_1,\ell'_1), g',(\ell_2,\ell'_2))$}{$V_\mathrm{diff}$}
			\ELSE
				~\postProcess{$((\ell_1,\ell'_1), g',(\ell_1,\ell'_2))$}{$\modified((\ell_1,\ell'_1))\cup\writ(op')$}
			\ENDIF
    \ELSE[$op'$ assume]
		    \STATE $\ell_p:=\ell_1;$ $V_\mathrm{diff}\assign\modified((\ell_1,\ell'_1));$
				  \WHILE{$\ell_p\neq\ell_\mathrm{err} \wedge \exists (\ell_p, op_\mathrm{a}, \ell_s)\in G: op_\mathrm{a}$ is assignment}
					  \STATE $\ell_p:=\ell_s;$ $V_\mathrm{diff}\assign V_\mathrm{diff}\cup\writ(op_\mathrm{a});$
					\ENDWHILE
				\IF{$\ell_p\neq\ell_\mathrm{err} \wedge (\exists (\ell_p,op,\ell_2) \in G: op=op'\vee ((op'\Rightarrow op) \wedge op'\not\equiv false)$)}
				  \IF{$V_\mathrm{diff}\cap (\rd(op)\cup\rd(op'))\neq\emptyset$}
						\postProcess{$((\ell_1,\ell'_1), g',\ell'_2)$}{$\emptyset$}
			    \ELSE
						~\postProcess{$((\ell_1,\ell'_1), g',(\ell_2, \ell'_2))$}{$V_\mathrm{diff}$}
					\ENDIF
				\ELSE				 
							~\postProcess{$((\ell_1,\ell'_1), g', (\ell_p, \ell'_2))$}{$V_\mathrm{diff}$}
			  \ENDIF
    \ENDIF
  \ENDFOR
\ENDWHILE
\IF{$\ell'_0\in\bad \vee (\ell'_0=\ell'_\mathrm{err}\wedge\ell_0\neq\ell_\mathrm{err})$}
  \textbf{return} $(\{\ell'_0\},\emptyset, \ell'_0,\{\ell'_0\})$
\ENDIF
\STATE $E\assign\{(n,g',n')\in \visited\times G'\times(\visited\cup\bad)\mid (n,g,n')\in\edges\vee (n'\in\bad\wedge\exists (n,g,(\cdot,n'))\in\edges)\};$
\STATE \textbf{return} $(\visited\cup\bad, E, (\ell_0, \ell'_0), \bad)$
\Statex
\STATE \postProcess{$(pred, g', succ)$}{$V_\mathrm{diff}$}
\IF{$succ\in L'$}
  \STATE $\bad\assign\bad\cup\{succ\};$ $\edges\assign\edges\cup\{(pred, g', succ)\};$ 
	\STATE $\visited\assign \visited\setminus\{(\cdot, succ)\in \visited\};$ $\waitlist\assign \waitlist\setminus\{(\cdot, succ)\in \waitlist\};$
\ELSE
	\IF{$succ\in (L\setminus\{\ell_\mathrm{err}\})\times \{\ell'_\mathrm{err}\}$}
	   \postProcess{$(pred, g', \ell'_\mathrm{err})$}{$\emptyset$}
  \ELSE
	  \IF{$succ \notin L\times\Delta \wedge succ\notin\{\ell_\mathrm{err}\}\times L' \wedge (succ\notin\visited\vee V_\mathrm{diff}\not\subseteq\modified(succ))$}
		  \STATE $\waitlist\assign \waitlist\cup\{succ\};$
    \ENDIF	
		\IF{$succ\notin L\times\Delta$}
		  \STATE $\visited\assign \visited\cup\{succ\};$ $\modified(succ)\assign \modified(succ)\cup V_\mathrm{diff};$
		\ENDIF
		 \STATE $\edges\assign\edges\cup\{(pred, g', succ)\};$
	\ENDIF	
\ENDIF
\end{algorithmic}

\end{algorithm}

Algorithm~\ref{alg:diffcompprop} shows  how our difference detector \diffDetectOurs works. 
It gets the original program~\(P\) and the modified program~\(P'\) as inputs.
We assume that both programs are given as deterministic CFAs.
When Alg.~\ref{alg:diffcompprop} finishes, it outputs its computed overapproximation of \(paths^\mathrm{rb}(P, P')\) as a difference graph.
The returned difference graph provides a witness path for each execution path of the modified program that is not contained in the overapproximation.

The outer while loop of Alg.~\ref{alg:diffcompprop} represents the algorithm's main part, which detects and simultaneously records witness paths~\(p\).
To this end, Alg.~\ref{alg:diffcompprop} it performs a step-by-step exploration of the paths of the modified program.
The step-by-step exploration also aims to align the paths to witness paths of the original program.
Thereby, it uses pairs~\((\ell,\ell')\in L\times L'\) of program locations as alignment points that describe where \(p\) and \(p'\) align. 
Algorithm~\ref{alg:diffcompprop} stores the detected alignment points in \(\visited\) and their connection in \(\edges\).
Furthermore, it uses data structure~\(\waitlist\) to keep track of which alignment points need to be explored and, thus, organizes the path exploration.
If a witness path~\(p\) aligned to a prefix of a path~\(p'\) of the modified program cannot be extended to that step of path~\(p'\) that follows after the prefix, i.e., no witness path can be constructed for \(p'\),
 Alg.~\ref{alg:diffcompprop} stops the exploration of \(p'\) and stores its next location in \(\bad\).
The set~\(\bad\) represents the set of regression bug indicator nodes of the difference graph, which characterize paths of the overapproximation.
To detect (dis)agreements on executability between paths and witness path extensions, data structure~\(\modified\) tracks for all aligned locations~\((\ell,\ell')\) the change affected variables, i.e., 
those variables which may have different values at alignment point~\((\ell,\ell')\). 

When the detection and recording of witness paths has finished, lines~19--21 of Alg.~\ref{alg:diffcompprop} build the difference graph.
The algorithm distinguishes two cases.
First, if \(\ell'_0\) is contained in \(\bad\) or \(\ell'_0\) is a new error location, no witness paths were detected because the alignment of the initial locations already failed.
Hence, all execution paths of the modified program  may cause regression bugs.
Therefore, the difference graph returned in line~19 consists of a single regression bug indicator node.
Second, we consider all other cases, in which Alg.~\ref{alg:diffcompprop} computed at least some witness paths.
To encode the witness paths into a difference graph, Alg.~\ref{alg:diffcompprop} uses the data structures \(\visited, \bad\) and \(\edges\), in which it records for each syntactic path either the information about its witness path or that no witness path has been detected.
During the construction of the difference graph, Alg.~\ref{alg:diffcompprop} basically copies the information of paths with and without witness paths into the difference graph.  
Hence, set~\(\bad\) becomes the regression bug indicator nodes.
The graph's nodes are the union of the determined aligned locations~\(\visited\) and the regression bug indicator nodes~\(\bad\).
The graph edges are derived from the connection between determined aligned locations and regression bug indicator nodes, thereby considering that aligned locations may be replaced by regression bug indicator nodes, which are not explored, i.e., they must not have successors.

Next, let us have a closer look at the path exploration with integrated witness path alignment. 
It starts in line~1, aligning the initial location of the two programs (initial alignment point of any path and its witness path) and with an empty set of connections (\(\edges\assign\emptyset;\)).
Also at the initial locations no modifications have been seen yet, i.e., we have not detected that any path for which the construction of a witness path failed (\(\bad\assign\emptyset\)) nor have we obtained knowledge about change affected variables (i.e., \(\modified\) does not track any change affected variable). 
Thereafter, each iteration of the outer while loop explores one aligned pair~\((\ell_1,\ell'_1)\) that has not yet been explored or whose change affected variables altered.
Thereby, it computes one successor per outgoing CFA edge~\(g'\) in the modified program.
The successor and the update of its change affected variables depend on \(g'\)'s operation.

In case of an assignment (lines~5--10), we ignore executability (assignments are always executable) and only check whether an outgoing edge with an identical assignment exists in the original program.
If an identical edge exists, we explore in lockstep and update the change affected variables according to the data dependencies.
For example, given aligned pair~\((\ell_2,\ell'_2)\) and outgoing CFA edge \((\ell'_2,r=-x;,\ell'_3)\), Alg.~\ref{alg:diffcompprop} detects the edge \((\ell_2,r=-x;,\ell_3)\) of the original program, which has an identical assignment, aligns the successors \((\ell_3,\ell'_3)\), and propagates that variable~\(r\) is no longer affected by a modification.
If no identical edge exists, we detect a modification.
Since it is non-trivial to detect a good alignment of original and modified program~\cite{Score}, we decided to (a)~always handle the modification like a new assignment (i.e., no step in the original program and all variables~\(\writ(op')\) become change affected) and (b)~if this assumption was wrong to postpone resynchronization to the beginning of assume operations.\footnote{To be more precise, our implementation also keeps the alignment if an assignment exists that changes the same variables (left out for simplicity of Alg.~\ref{alg:diffcompprop}).}
As an example for a modification, let us consider the aligned pair~\((\ell_0,\ell'_0)\) and corresponding outgoing edge  \((\ell'_0,r=x;,\ell'_1)\).
Algorithm~\ref{alg:diffcompprop} detects the modification and only steps forward in the modified program, thus, aligns \((\ell_0,\ell'_1)\).
Additionally, it propagates that variable~\(r\) is affected by the modification.
In case of a match as well as a modification, procedure~\textsc{process} updates the data structures based on the suggested alignment and determined change affected variables.

In case of an assume operation (lines~11--18), we account for the mismatch of assignments and first perform any postponed synchronization (lines~12--14).
Thereby, we explore all assignments in the original program  until we hit the next assume operation or an error location (i.e., no regression bug possible).
Explored assignments are treated like modifications, i.e., we add their modified variables to the changed affected variables~\(V_\mathrm{diff}\) of the alignment. 
After resynchronization, line~15 checks whether both programs provide matching assume operations.
Assume operations match if they are identical or the original assume is more general\footnote{To get deterministic difference graphs for our proofs, \(op'\) must be satisfiable, too.}. 
If a matching assume operation is found, we continue with the executability check (line~16).
We can only safely guarantee executability if the assume operations are affected by the modification.
If the matching check fails, we consider the assume of the modified program as a new operation.
Note that the assume operation may restrict the executability of the modified path and, thus, the relevant inputs.
However, it does not restrict the executability of the witness path on relevant inputs. 
Hence, only if matching assumes are affected by the modification (line~16), we failed to detect a witness path and stop the alignment.
To this end, we let the alignment end in the location of the modified path, which describes a regression bug indicator node and for which change affected variables are irrelevant.
In all other cases (lines~17 and 18), we align original and modified program accordingly and since assume operations do not change the data state, forward the already determined change affected variables~\(V_\mathrm{diff}\).  
Again, procedure~\textsc{process} updates the data structures.
As an example, let us consider aligned pair~\((\ell_0,\ell'_1)\), outgoing CFA edge \((\ell'_1,x>0,\ell'_2)\) and \(\modified((\ell_0,\ell'_1))=\{r\}\).
Algorithm~\ref{alg:diffcompprop} first explores the edge \((\ell_0, r=-x;,\ell_1)\) of the original when performing resynchronizes and thereby adds \(r\) to the change affected variables~$V_\mathrm{diff}$. 
Then, it finds a matching assume edge \((\ell_1, x>0, \ell_2)\) in the original program. 
Since both assume operation do not consider any change affected variables ($V_\mathrm{diff}=\{r\}$), Alg.~\ref{alg:diffcompprop} aligns \((\ell_2,\ell'_2)\) and propagates \(V_\mathrm{diff}\).

Lastly, let us have a look at the procedure~\textsc{process}, which is responsible for updating the data structures.
To perform the update, it gets the proposed successor (edge) and the determined change affected variables.
In case the witness path extension failed (lines~22--24), we stop exploration, add the successor to \(\bad\), add the successor edge, and remove any aligned locations considering the successor's location form \(\visited\) and \(\waitlist\).
The edges are adapted later in line~19.
In case alignment (extension) has been proposed (lines~26--32), we first need to account for violations by paths of the modified program.
If only the path of the modified program violates the property, we recursively call \textsc{process}, but now exchanging the alignment with the location of the modified path, a regression bug indicator node.\footnote{To further exclude non-regression bugs, our implementation first searches for an error location of the original program in its follow-up sequence of assignments.}
Otherwise, line~28--29 mark the successor for (re)exploration if we neither end in a regression bug indicator node (difference already established), nor align with an error location of the original program (regression bug impossible), nor the successor and all change affected variables are known (information already considered for exploration). 
If we do not end in a regression bug indicator node, we need to ensure that successor information and the detected change affected variables are available for exploration and graph construction. 
Hence, lines~30--31 add them to \(\visited\) and \(\modified\).
Finally, line~32 adds the corresponding edge to \(\edges\).

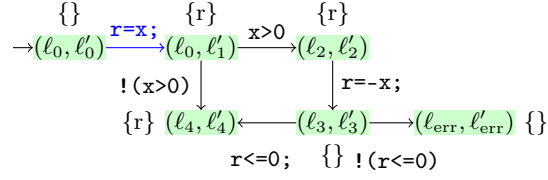
\begin{figure}[t]
\centering
\begin{tikzpicture}
		\footnotesize
		\node[ inner sep=0cm,fill=green!20] (l0) at (-1.75,1) {$(\ell_0,\ell'_0)$};
		\node[ inner sep=0cm,fill=green!20] (l1) at (0,1) {$(\ell_0,\ell'_1)$};
		\node[ inner sep=0cm,fill=green!20] (l2) at (1.75,1) {$(\ell_2, \ell'_2)$};
		\node[ inner sep=0cm,fill=green!20] (l3) at (1.75,0) {$(\ell_3,\ell'_3)$};
		\node[ inner sep=0cm,fill=green!20] (l4) at (0,0) {$(\ell_4,\ell'_4)$};
		\node[ inner sep=0cm,fill=green!20] (l5) at (3.5,0) {$(\ell_\mathrm{err},\ell'_\mathrm{err})$};
		\path[->,draw] (-2.5,1) -- (l0);
		\path[->,draw,blue] (l0) -- node[above,font=\footnotesize]{\texttt{r=x;}} (l1);
		\path[->,draw] (l1) -- node[above,font=\footnotesize]{{\texttt{x>0}}} (l2);
		\path[->,draw] (l1) -- node[left,font=\footnotesize]{{\texttt{!(x>0)}}} (l4);
		\path[->,draw] (l2) -- node[right,font=\footnotesize]{\texttt{r=-x;}} (l3);
		\path[->,draw] (l3) -- node[sloped,below, yshift=-2ex, xshift=-0.5ex,font=\footnotesize]{\texttt{r<=0;}} (l4);
		\path[->,draw] (l3) -- node[sloped, below, yshift=-2ex, xshift=0.5ex,font=\footnotesize]{\texttt{!(r<=0)}} (l5);
		\node at (l0.north) [anchor=south, font=\footnotesize] {\{\}};
		\node at (l1.north) [anchor=south, font=\footnotesize] {\{r\}};
		\node at (l2.north) [anchor=south, font=\footnotesize] {\{r\}};
		\node at (l3.south) [anchor=north, font=\footnotesize] {\{\}};
		\node at (l4.west) [anchor=east, font=\footnotesize] {\{r\}};
		\node at (l5.east) [anchor=west, font=\footnotesize] {\{\}};
	\end{tikzpicture}
\caption{Difference graph of Alg.~\ref{alg:diffcompprop} for original and modified CFA from Fig.~\ref{fig:exampleprograms}}
\label{fig:exampleDG}
\end{figure}

Next, we demonstrate Alg.~\ref{alg:diffcompprop} on our example from Fig.~\ref{fig:exampleprograms}.
Figure~\ref{fig:exampleDG} shows the computed difference graph.
Algorithm~\ref{alg:diffcompprop} starts with the pair of initial locations~\((\ell_0,\ell'_0)\) and no change affected variables. 
Next, it detects the modified assignment, which is handled as a new statement, and, thus, only moves forward in the modified program assuming the the assigned variable~\(r\) to be affected.
Thereafter, Alg.~\ref{alg:diffcompprop} explores the first two assume operations of the modified program.
For both assume operation, it first performs the resynchronization that takes the first assignment of the original program into and adds its assigned variable~\(r\) to the changed affected variables~\(V_\mathrm{diff}=\{r\}\).
Then, it detects the identical assume operations in the original program, which do not consider any changed affected variables.
Hence, Alg.~\ref{alg:diffcompprop} propagates the changed affected variables~\(\{r\}\), and continues exploration.
Afterwards, it explores the second assignment of the modified program.
It detects an identical assignment in the original program and removes \(r\) from the change affected variables because \(r\) is redefined and no longer affected.
Finally, Alg.~\ref{alg:diffcompprop} explores the assume operations resulting from the assertion.
For both, there exists an identical assume operation in the original program, which is not affected by the modification.
Thus, the assumptions are explored in lockstep and the change affected variables are forwarded.

Our goal is to use Alg.~\ref{alg:diffcompprop} as a difference detector for difference detection with conditions.
To keep soundness of difference detection, we hence need to show that Alg.~\ref{alg:diffcompprop} constructs proper difference graphs.
Theorem~\ref{theo:detector}, whose proof is available in the appendix, claims this.
\begin{theorem}\label{theo:detector}
Algorithm~\ref{alg:diffcompprop} returns a difference graph.
\end{theorem}

\section{Evaluation}
In our evaluation, we aim to compare difference verification with our difference condition extractor (\diffDetectOurs plus Alg.~\ref{alg:cond}) against a full verification and against difference verification with the syntax-based difference condition extractor~\diffCond from~\cite{DiffCond} in terms of efficacy (i.e., number of solved verification tasks) and efficiency (i.e., CPU time).

\subsection{Experimental Setup}

\textbf{Verification Tasks.}
We use two sets of verification tasks in our evaluation: a set of combination tasks and a set of regression tasks.
Each task in any of the two sets consist of a original and modified program, which both consider the property that the function \texttt{reach\_error()} is never called.
Furthermore, at most one of the two programs in each task may actually call the function \texttt{reach\_error()} and, thus, violate the property.
This is important to ensure that verification and difference verification can only detect regression bugs.

The set of combination tasks contains the \num{10426}~ tasks already utilized in Beyer et al.'s evaluation of difference verification with conditions~\cite{DiffCond}
The tasks of this set are organized in five categories (1)~\texttt{eca05+token}, (2)~\texttt{gcd+newton}, (3)~\texttt{pals+eca12}, (4)~\texttt{sfifo+token}, (5)~\texttt{square+softflt} and their programs are the result of combining two verification tasks.

The set of regression tasks were derived by Beyer et al.~\cite{PrecisionReuse} from
revisions of \num{62}~Linux device drivers and used to evaluate their incremental verification technique. 
Each program in the task set relates to one revision of a device driver and examines one of six API usage rules. The rule examination is woven into the program s.t.\ a rule violation leads to a call of the function \texttt{reach\_error()}. 
Furthermore, original and modified program examine the same rule and the modified program stems from the next revision available for which the ground truth is known
Nevertheless, original and modified program typically differ in several hundred lines.
 and the ground truth for the modified program is known. 
The resulting task set contains \num{3936}~so-called regression tasks.

\textbf{Verifiers.}
For our experiments, we select two different off-the-shelf verifiers from the SV-COMP~2022 participants based on the performance of the SV-COMP~2022 participants in the categories \texttt{Reach\-Safety-Combinations} and \texttt{Software\-Systems-DeviceDriversLinux64-ReachSafety}, the categories closest to our tasks.
More concretely, we select the best tool~\cpacheckerSVCOMP and \esbmc, the next best tool that uses a different verification platform. 
\cpacheckerSVCOMP runs a feature-based strategy selection~\cite{CPAcheckerStrategy} to select the most suitable combination of sequential analyses.
These combination may e.g., apply bounded, explicit, or predicate model checking or k-induction.
\esbmc~\cite{ESBMC-SVCOMP} applies k-induction enhanced with invariants. 
We add to these two off-the-shelf verifiers, which are no conditional verifiers, i.e., their difference verifiers will need to use reducer-based conditional verifier approach~\cite{ReducerCMC}, the native conditional verifier \predicate~\cite{CMC} that combines predicate model checking with adjustable block-encoding~\cite{ABE}, abstracts at loop heads and applies counterexample-guided abstraction refinement 
with lazy refinement~\cite{LazyAbstraction} to compute relevant predicates.

We use all three verifiers~\(V\) standalone for full verification, i.e., they verify the complete modified program, and as verifier component in difference verifiers.
For the difference verifiers, we use reducer-based CMC~\cite{ReducerCMC} to combine every verifier~\(V\)  with the syntax-based difference condition extractor \diffCond from~\cite{DiffCond}~(\(V^{\diffSynRed}\)) and our 
difference condition 
 extractor (\(V^{\diffOursRed}\)).
Also, we natively combine the conditional verifier~\predicate with the syntax-based extractor \diffCond~(\predicateDiffSynNative) and our extractor~(\predicateDiffOursNative).

For the two off-the-shelf verifiers,
we use their SV-COMP 2022~\cite{SVCOMP22} version, while for \predicate, the reducer, and the two difference condition extractors we use the more recent 
\cpachecker version~\num{40843}, which includes all four components.

\textbf{Computing Environment.}
We run our experiments on machines with an Intel Xeon E3-1230 v5 CPU (frequency of 3.4\,\GHz, 8 processing units) and 33\,GB of memory.
All machines run an Ubuntu 20.04 (Linux kernel~5.4.0).
Furthermore,  we restrict each verifier runto 4~processing units, 15\,$\min$ of CPU time, and 15\,GB of memory using \benchexec~\cite{BenchExec} version~3.11.

\subsection{Experimental Results}

\textbf{RQ~1 Is Difference Verification with Our Extractor Better Than a Complete Verification?}
The evaluation of Beyer et al.~\cite{DiffCond} revealed that for a significant number of tasks difference verification with the syntax-based difference condition extractor~\diffCond is more effective or more efficient than verifying the complete modified program of the task. 
With our first research question we want to find out if these results carry over to difference verification with our new difference condition extractor.

\begin{table}[t]%
\caption{Number of correct proofs (\cproof), of correct alarms (\calarm), and of incorrectly solved tasks (\iresult) detected by each verifier and their related difference verifiers, which use our difference condition extractor}
\label{tab:effectivityNumbers-comp-verify}
\centering
\scalebox{0.8325}{
\begin{tabular}{l  >{\columncolor{black!10}[\tabcolsep]}d{0} >{\columncolor{black!10}[\tabcolsep]}d{0} >{\columncolor{black!10}[\tabcolsep]}d{0}  d{0} d{0} d{0}  >{\columncolor{black!10}[\tabcolsep]}d{0} >{\columncolor{black!10}[\tabcolsep]}d{0} >{\columncolor{black!10}[\tabcolsep]}d{0}  d{0} d{0} d{0}  >{\columncolor{black!10}[\tabcolsep]}d{0} >{\columncolor{black!10}[\tabcolsep]}d{0} >{\columncolor{black!10}[\tabcolsep]}d{0}  d{0} d{0} d{0} }

 &  \multicolumn{3}{>{\columncolor{black!10}[\tabcolsep]}c}{\smaller eca05+token} & \multicolumn{3}{c}{\smaller gcd+newton} & \multicolumn{3}{>{\columncolor{black!10}[\tabcolsep]}c}{\smaller pals+eca12} & \multicolumn{3}{c}{\smaller sfifo+token} & \multicolumn{3}{>{\columncolor{black!10}[\tabcolsep]}c}{\smaller square+softflt} & \multicolumn{3}{c}{\smaller regression}\\
 & \multicolumn{3}{>{\columncolor{black!10}[\tabcolsep]}c}{\smaller (\num{2340}+\num{1300})} & \multicolumn{3}{c}{\smaller (\num{1352}+\num{572})} & \multicolumn{3}{>{\columncolor{black!10}[\tabcolsep]}c}{\smaller (\num{1700}+\num{1050})} & \multicolumn{3}{c}{\smaller (\num{1206}+\num{663})} & \multicolumn{3}{>{\columncolor{black!10}[\tabcolsep]}c}{\smaller (\num{165}+\num{75})} & \multicolumn{3}{c}{\smaller (\num{3936}+\num{0})}\\
 & \multicolumn{1}{>{\columncolor{black!10}[\tabcolsep]}c}{\cproof} & \multicolumn{1}{>{\columncolor{black!10}[\tabcolsep]}c}{\calarm} & \multicolumn{1}{>{\columncolor{black!10}[\tabcolsep]}c}{\iresult}  & \multicolumn{1}{c}{\cproof} & \multicolumn{1}{c}{\calarm} & \multicolumn{1}{c}{\iresult} & \multicolumn{1}{>{\columncolor{black!10}[\tabcolsep]}c}{\cproof} & \multicolumn{1}{>{\columncolor{black!10}[\tabcolsep]}c}{\calarm} & \multicolumn{1}{>{\columncolor{black!10}[\tabcolsep]}c}{\iresult}  & \multicolumn{1}{c}{\cproof} & \multicolumn{1}{c}{\calarm} & \multicolumn{1}{c}{\iresult} & \multicolumn{1}{>{\columncolor{black!10}[\tabcolsep]}c}{\cproof} & \multicolumn{1}{>{\columncolor{black!10}[\tabcolsep]}c}{\calarm} & \multicolumn{1}{>{\columncolor{black!10}[\tabcolsep]}c}{\iresult} & \multicolumn{1}{c}{\cproof} & \multicolumn{1}{c}{\calarm} & \multicolumn{1}{c}{\iresult}\\
\toprule
\predicate & 912 & \textbf{1040} & 0 & 0 & 520 & 0 & 0 & 50 & 0 & 558 & \textbf{507} & 0 & 22 & 51 & 0 & 2595 & 0 & 0  \\ 
\rowcolor{yellow!10}\predicateDiffOursNative & \textbf{1400} & 985 & 0 & \textbf{156} & 520 & 0 & \textbf{125} & 75 & 0 & 594 & 488 & 0 & 115 & \textbf{75} & 0 & 2663 & 0 & 0\\
\rowcolor{yellow!10}\predicateDiffOurs & 1390 & 955 & 0 & \textbf{156} & \textbf{572} & 0 & \textbf{125} & 50 & 0 & \textbf{639} & 456 & 0 & \textbf{135} & 60 & 0 & \textbf{2685} & 0 & 0\\
\midrule
\cpachecker & 633 &  \textbf{1296} & 0 & 0 & 494 & 0 & 0 & 100 & 0 & 434 & 510 & 0 & 0 & \textbf{75} & 0 & \textbf{3931} & 0 & 0 \\
\rowcolor{yellow!10}\cpacheckerDiffOurs & \textbf{1119} & 1252& 0 & \textbf{156} & \textbf{520} & 0 & \textbf{671} & \textbf{500} & 0 & \textbf{480} & \textbf{638} & 0 & \textbf{117} & 60 & 0 & 3618 & 0 &0 \\
\midrule
\esbmc & 0 & \textbf{1125} & 0 & 570 & 572 & 0 & \textbf{198} & \textbf{530} & 0 & 0 & \textbf{663} & 0 & \textbf{165} & 75 & 0 & 0 & 0 & 0 \\
\rowcolor{yellow!10}\esbmcDiffOurs & 0 & 900 & 0 & \textbf{880} & 572 & 0 & 0 & 70 & 10 & \textbf{101} & 618 & 0 & 156 & 75 & 0 & 0 & 0 & 0 \\ 
\bottomrule
\end{tabular}
}
\end{table}

We start with examining effectiveness.
Table~\ref{tab:effectivityNumbers-comp-verify} shows per task category the number of correctly proved~(\cproof) and disproved~(\calarm) tasks as well as the number of incorrectly solved tasks (\iresult) by a verifier and its difference verifier that uses our difference condition extractor. Note that all incorrectly solved tasks are false proofs. The total number of tasks with and without property violations in category are shown at the top of the table. 

When comparing the number of correctly proved tasks~(\cproof), we observe that for each verifier, there exist categories for which the difference verifier (highlighted in yellow) has a higher number, i.e., correctly proves more tasks. 
Often there, also exist categories in which the difference verifier correctly finds more property violations (\calarm). 
In some cases, the difference verifier succeeds while the full verification runs out of resources.
In other cases, difference verification excludes paths with program features not supported by the underlying verifier.
In summary, difference verification with our difference condition extractor can be more effective than a full verification but it does not dominate the full verification.
For \predicate and \cpachecker our difference verification approach is rarely less and often more effective.
However, for \esbmc the difference verifier only performs better in some cases and even misses \num{10}~bugs. Since \cpacheckerDiffOurs correctly detected those \num{10}~bugs, the bugs seem to be an issue of \esbmc.
One important reason for a worse performance of the difference verifier,  which Beyer et al.~\cite{DiffCond} also reported, is the residual program.
The residual program is generated and then verified by the reducer-based conditional verifier and often changes the program structure.
For example, we observed that the residual programs may contain deeply nested branching structures that regularly exceeded the maximum bracket nesting level of the \esbmc parser.
Also, residual programs heavily use goto statements, but no while or for loop constructs, which makes loop detection more difficult for \esbmc.
However, the issues with the residual program are orthogonal to difference verification with conditions. In particular, they do not stem from difference verification with conditions, but how the reducer is implemented and how non-conditional verifiers are transformed into conditional verifiers.

\begin{figure}[t]
	\centering
	  \includegraphics[scale=0.35]{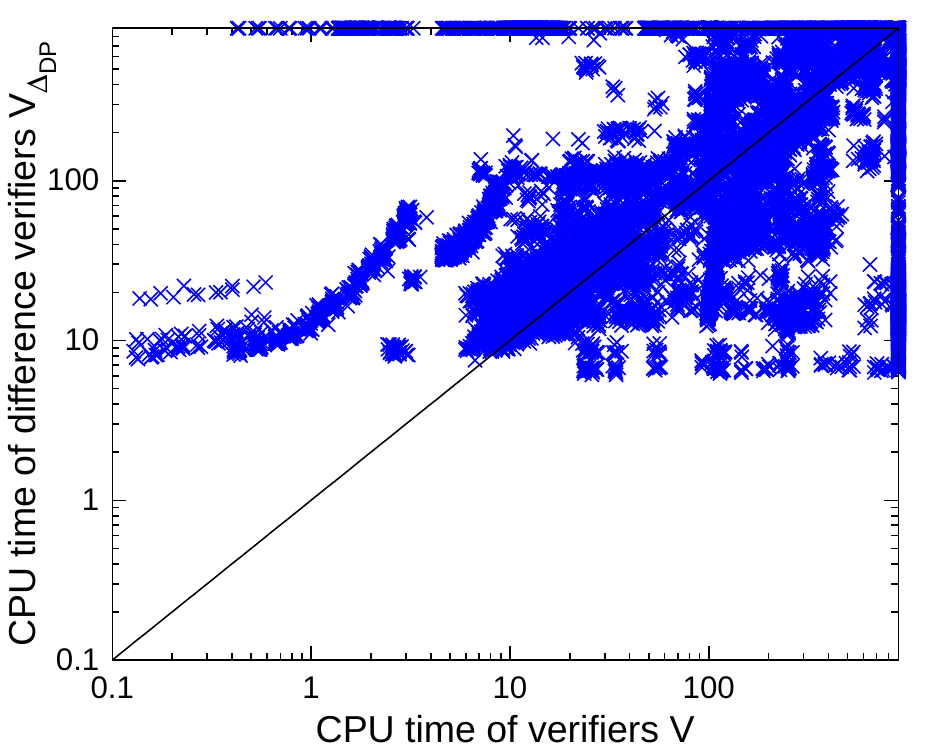}
	\caption{CPU time for verification (x-axis) vs.\ CPU time for difference verification with our extractor (y-axis)}
	\label{fig:VvsDiffV}
\end{figure}

Next, we continue our study with the examination of efficiency.
To this end, we compare the runtime of the complete verification with the runtime of the difference verifier.
We measure runtime in CPU time.
Figure~\ref{fig:VvsDiffV} shows a scatter plot that compares for all verifiers the CPU time the verifier uses to analyze task~\(t\) (x-axis) with the CPU time the respective difference verifier requires to inspect the same task~\(t\) (y-axis). In case a task is not solved, we use the maximum CPU time of \num{900}\,s.
Note that for \esbmc we exclude the 10~tasks that the difference verifier solved incorrectly, i.e., the scatter plot does not consider times of incorrectly solved tasks.
We observe that a significant number of data points (40\%) in the scatter plot are below the diagonal, i.e., for those task-verifier pairs the difference verification is faster than the complete verification of the modified program.
In further cases (about 20\%), the conditional verifier is still faster than complete verification, but conditional verification plus difference condition extracting takes longer than the complete verification.
A detailed inspection reveals that for some tasks categories like \texttt{square+softflt} the costs of the extractor can be compensated by the saved time of the conditional verifier, while this hardly the case for tasks of the \texttt{pals+eca12} category. We also observe differences between verifiers.

\llbox{
To sum up, our results confirm the observations of Beyer et al~\cite{DiffCond}.
Also, difference verification with our condition extractor may be more effective and efficient than complete verification.
}

\begin{figure}[t]
	\centering
		\subfloat[\#condition~states generated by syntax-based (x-axis) and by our (y-axis) difference condition extractor  \label{fig:diffExtractorCompare:states}]{\includegraphics[scale=0.35]{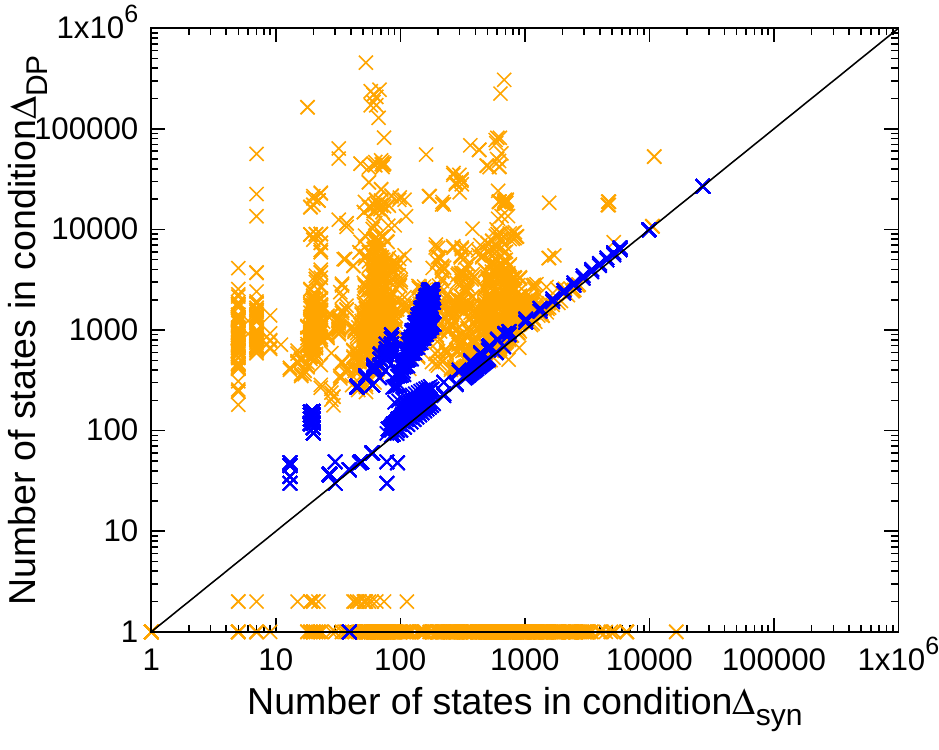}}
		\hfill
		\subfloat[CPU time of syntax-based difference condition extractor (x-axis) and of our difference condition extractor (y-axis)\label{fig:diffExtractorCompare:time}]{\includegraphics[scale=0.35]{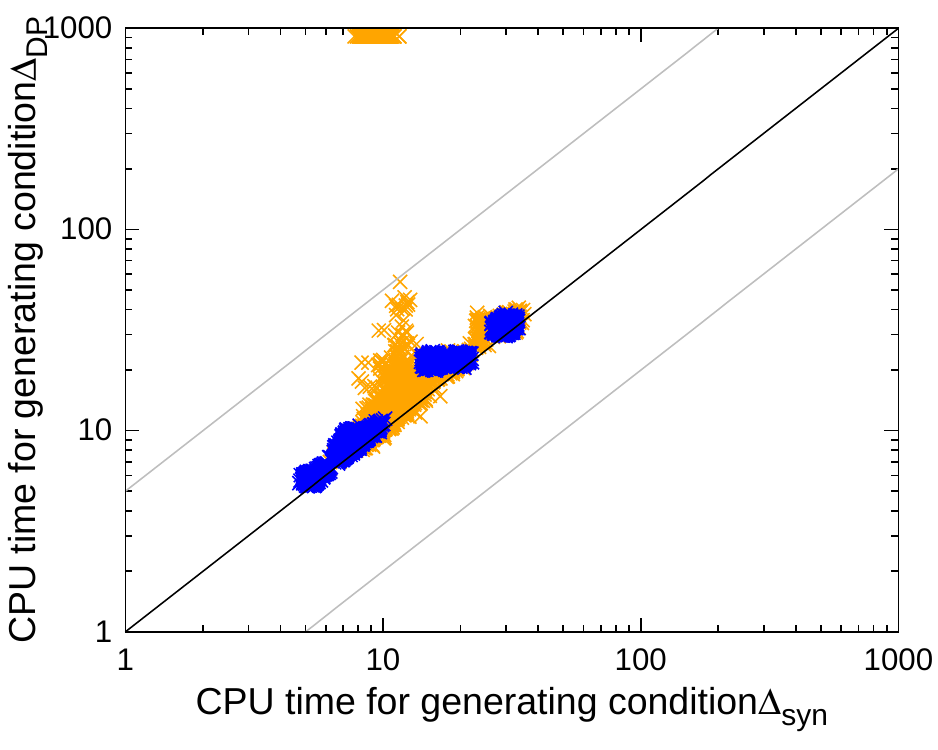}}
	\caption{Comparison of the two difference condition extractors}
	\vspace{-1em}
\end{figure}

\textbf{RQ~2 Do \diffCond and Our Difference Condition Extractor Differ?}
The previous research question has shown that difference verification with our extractor can be beneficial. 
Since this is also true for difference verification with the syntax-based extractor~\diffCond~\cite{DiffCond}, we may speculate that this is only because there is no difference in the extractors.
This research question investigate whether this hypothesis is true.

Difference condition extractors may affect the performance of the difference verifier in two ways, namely if they~(a)~generate different conditions or (b)~differ in performance.
We start with examining whether the generated conditions differ.
We forego a semantic difference and focus on a syntactical difference of conditions, which is easier to investigate. Note that this is possible because all our difference verifiers unroll the program along the condition.
Hence, a difference in syntax, especially in structure, already affects the  performance of the difference verifier.
As a proxy for structural difference, we use the number of states in the condition.
Figure~\ref{fig:diffExtractorCompare:states} shows a scatter plot that compares for each task (except for \num{286}~regression tasks for which our extractor failed mostly due to timeout) the number of states in the condition generated by the syntax-based difference condition extractor \diffCond (\(\diffSynUn\), x-axis) with the number of states in the condition produced by our difference condition extractor~(\(\diffOursUn\), y-axis).
We note that most data points in Fig.~\ref{fig:diffExtractorCompare:states} are above or below the diagonal, i.e., the conditions differ in the number of states.
A detailed inspection of our data reveals that only for \num{317}~tasks  (<\num{5}\% of all tasks) the condition size is the same. 
In addition, the conditions generated by our extractor are larger for
While for \num{82}\% 
of the combination tasks (blue data points) and for \num{55}\%  of the regression tasks (orange data points). We note that the generated conditions may be larger by several magnitudes.
However,there also exist \num{1676}~data points on the x-axis, i.e., our extractor generates conditions with a single accepting state and, hence, proved the absence of regression bugs. In contrast, the syntax-based difference verification extractor only generates conditions with a single accepting state for \num{75}~tasks, for which our extractor produces the same condition.
In summary, we conclude that our extractor produces more precise conditions than \diffCond.

We continue to investigate the performance difference~(b) in terms of runtime of the difference condition extractors.
As in the previous research question, we compare runtime measured in CPU time.
Figure~\ref{fig:diffExtractorCompare:time} shows a scatter plot that compares for each task the CPU time taken by \diffCond~(\(\diffSynUn\), x-axis) with the CPU time needed by our extractor~(\(\diffOursUn\), y-axis).
We note that most of the data points are above the diagonal, i.e., our extractor takes longer. An analysis of the data revealed that our extractor takes up to 5-times longer with a median and mean increase between \num{8}-\num{40}\% per category. Furthermore, it even times out for \num{285}~regression tasks.
The increase could have been expected considering that our constructor produces larger and more precise conditions.
Nevertheless, the runtime of our difference condition extractor remains mostly reasonable, i.e., it typically takes less than \num{10}\% of the runtime limit of the difference verifier (\num{900}\,s).

\llbox{
In summary, our extractor computes more precise conditions, but is slower.
Hence, we confirm the difference of the two difference condition extractors.
}

\begin{table}[t]%
\caption{Number of correct proofs (\cproof), of correct alarms (\calarm), and of incorrectly solved tasks (\iresult) detected by each verifier and their related difference verifiers}
\label{tab:effectivityNumbers-comp-diffV}
\centering
\scalebox{0.8325}{
\begin{tabular}{l  >{\columncolor{black!10}[\tabcolsep]}d{0} >{\columncolor{black!10}[\tabcolsep]}d{0} >{\columncolor{black!10}[\tabcolsep]}d{0}  d{0} d{0} d{0}  >{\columncolor{black!10}[\tabcolsep]}d{0} >{\columncolor{black!10}[\tabcolsep]}d{0} >{\columncolor{black!10}[\tabcolsep]}d{0}  d{0} d{0} d{0}  >{\columncolor{black!10}[\tabcolsep]}d{0} >{\columncolor{black!10}[\tabcolsep]}d{0} >{\columncolor{black!10}[\tabcolsep]}d{0}  d{0} d{0} d{0} }

 &  \multicolumn{3}{>{\columncolor{black!10}[\tabcolsep]}c}{\smaller eca05+token} & \multicolumn{3}{c}{\smaller gcd+newton} & \multicolumn{3}{>{\columncolor{black!10}[\tabcolsep]}c}{\smaller pals+eca12} & \multicolumn{3}{c}{\smaller sfifo+token} & \multicolumn{3}{>{\columncolor{black!10}[\tabcolsep]}c}{\smaller square+softflt} & \multicolumn{3}{c}{\smaller regression}\\
 & \multicolumn{3}{>{\columncolor{black!10}[\tabcolsep]}c}{\smaller (\num{2340}+\num{1300})} & \multicolumn{3}{c}{\smaller (\num{1352}+\num{572})} & \multicolumn{3}{>{\columncolor{black!10}[\tabcolsep]}c}{\smaller (\num{1700}+\num{1050})} & \multicolumn{3}{c}{\smaller (\num{1206}+\num{663})} & \multicolumn{3}{>{\columncolor{black!10}[\tabcolsep]}c}{\smaller (\num{165}+\num{75})} & \multicolumn{3}{c}{\smaller (\num{3936}+\num{0})}\\
 & \multicolumn{1}{>{\columncolor{black!10}[\tabcolsep]}c}{\cproof} & \multicolumn{1}{>{\columncolor{black!10}[\tabcolsep]}c}{\calarm} & \multicolumn{1}{>{\columncolor{black!10}[\tabcolsep]}c}{\iresult}  & \multicolumn{1}{c}{\cproof} & \multicolumn{1}{c}{\calarm} & \multicolumn{1}{c}{\iresult} & \multicolumn{1}{>{\columncolor{black!10}[\tabcolsep]}c}{\cproof} & \multicolumn{1}{>{\columncolor{black!10}[\tabcolsep]}c}{\calarm} & \multicolumn{1}{>{\columncolor{black!10}[\tabcolsep]}c}{\iresult}  & \multicolumn{1}{c}{\cproof} & \multicolumn{1}{c}{\calarm} & \multicolumn{1}{c}{\iresult} & \multicolumn{1}{>{\columncolor{black!10}[\tabcolsep]}c}{\cproof} & \multicolumn{1}{>{\columncolor{black!10}[\tabcolsep]}c}{\calarm} & \multicolumn{1}{>{\columncolor{black!10}[\tabcolsep]}c}{\iresult} & \multicolumn{1}{c}{\cproof} & \multicolumn{1}{c}{\calarm} & \multicolumn{1}{c}{\iresult}\\
\toprule
\predicateDiffSynNative & 1395 & \textbf{987} & 0 & 48 & 520 & 0 & 10 & 52 & 0 & 589 & 481 & 0 & 61 & {75} & 0 & 2599 & 0 & 0 \\
\rowcolor{yellow!10}\predicateDiffOursNative & \textbf{1400} & 985 & 0 & \textbf{156} & 520 & 0 & \textbf{125} & \textbf{75} & 0 & \textbf{594} & \textbf{488} & 0 & \textbf{115} & {75} & 0 & \textbf{2663} & 0 & 0\\
\midrule
\predicateDiffSyn & \textbf{1394} & \textbf{975} & 0 & 48 & 474 & 0 & 10 &  \textbf{95} & 0 & \textbf{642} & \textbf{468} & 0 & 81 & 45 & 0 & 2639 & 0 & 0  \\
\rowcolor{yellow!10}\predicateDiffOurs & 1390 & 955 & 0 & \textbf{156} & \textbf{572} & 0 & \textbf{125} & 50 & 0 & 639 & 456 & 0 & \textbf{135} & \textbf{60} & 0 & \textbf{2685} & 0 & 0\\
\midrule
\cpacheckerDiffSyn & 1000 & \textbf{1282} & 0 & 48 & 469 & 0 & 41 & 140 & 0 & {480} & \textbf{663} & 0 & 61 & 45 & 0 & \textbf{3906} & 0 & 0  \\
\rowcolor{yellow!10}\cpacheckerDiffOurs & \textbf{1119} & 1252& 0 & \textbf{156} & \textbf{520} & 0 & \textbf{671} & \textbf{500} & 0 & {480} & 638 & 0 & \textbf{117} & \textbf{60} & 0 & 3618 & 0 &0 \\
\midrule
\esbmcDiffSyn & 0 & 0 & 0 & 333 & 572 & 0 & 0 & 0 & 0 & 0 & 78 & 0 & 105 & 75 & 0 & 0 & 0 & 0\\
\rowcolor{yellow!10}\esbmcDiffOurs & 0 & \textbf{900} & 0 & \textbf{880} & 572 & 0 & 0 & \textbf{70} & 10 & \textbf{101} & \textbf{618} & 0 & \textbf{156} & 75 & 0 & 0 & 0 & 0 \\ 
\bottomrule
\end{tabular}
}
\end{table}

\textbf{RQ~3 Does Difference Verification with Our Extractor Improve over Difference Verification with the Syntax-Based Extractor~\diffCond?}
Based on the results of the previous research question, we can expect a performance difference when exchanging the syntax-based difference condition extractor with our extractor. However, we do not yet know whether the difference results in an improvement with respect to effectiveness or efficiency. 
Therefore, we investigate these two aspects in this research question.

Again, we start with studying the effectiveness.
For all difference verifiers considered in our experiments, Tab.~\ref{tab:effectivityNumbers-comp-diffV} shows the number of correctly proved~(\cproof) and disproved~(\calarm) tasks as well as the number of incorrectly solved tasks (\iresult).
The structure of the table is similar to the one of Tab.~\ref{tab:effectivityNumbers-comp-verify}. The only difference is that we exchanged the rows for the complete verification by the results for the difference verifiers using the syntax-based difference condition extractor. 
Now, let us compare the results of those difference verifiers that only differ in their difference condition extractor ({\diffSynRed} vs.\ {\diffOursRed} or {\diffSynNat} vs.\ \diffOursNat).
When comparing the difference verifiers using the native approach for difference verification with \predicate (first two rows), difference verification with our difference condition extractor (\predicateDiffOursNative) mostly solves more tasks correctly than difference verification with the syntax-based extractor~\diffCond (\predicateDiffSynNative).
A more detailed study revealed that using our difference condition extractor, which computes a more precise overapproximation of the paths with regression bugs, enables the difference verifier to solve additional tasks 
and only for the regression tasks there exist several hundred tasks solvable by \predicateDiffSynNative for which \predicateDiffOursNative fails due to timeout in the extractor.
When using the reducer-based approach for \predicate, 
results are more diverse.
For some categories, difference verification with our extractor~(\predicateDiffOurs) detects significantly fewer alarms (\calarm) and sometimes also a slightly smaller number of proofs (\cproof), which \predicateDiffSyn detects close to timeout.
Since most of these alarms are detected by \predicateDiffOursNative, we think the structure of the residual program  is responsible for the decrease in performance and not our extractor.
More importantly, there exist categories like \texttt{gcd+newton} and \texttt{square+softflt} in which our extractor enables the difference verifier to solve all tasks solved by \predicateDiffSyn plus some additional tasks.
Also, in case of the verifier~\cpachecker, the difference verifier using our extractor (\cpacheckerDiffOurs) typically detects additional proofs,  but sometimes detects fewer alarms than the difference verifier  \cpacheckerDiffSyn based on the syntax-based extractor.
While the decrease in category~\texttt{regression} is caused by the timeout of our extractor, a decrease of detected alarms occurs only in categories for which \predicateDiffOurs also performs worse, and is, thus, likely caused by the residual program.
For \esbmc, difference verification with our extractor (\esbmcDiffOurs) always improves over \esbmcDiffSyn.
Typically, \esbmcDiffOurs solves the same tasks as \esbmcDiffSyn as well as some additional tasks.
Note that both \esbmcDiffSyn and \esbmcDiffOurs do not solve any regression task because the entry function for analysis is not main, and, therefore not supported by \esbmc.

\begin{figure}[t]
	\centering
		\includegraphics[width=0.3\textwidth]{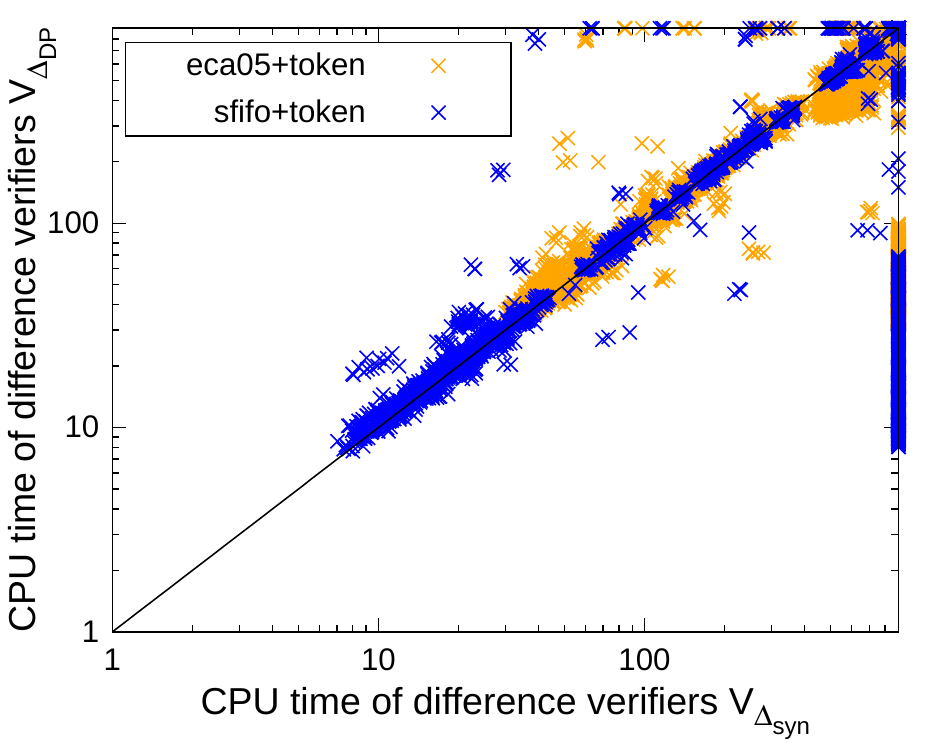}
		\hfill
		\includegraphics[width=0.3\textwidth]{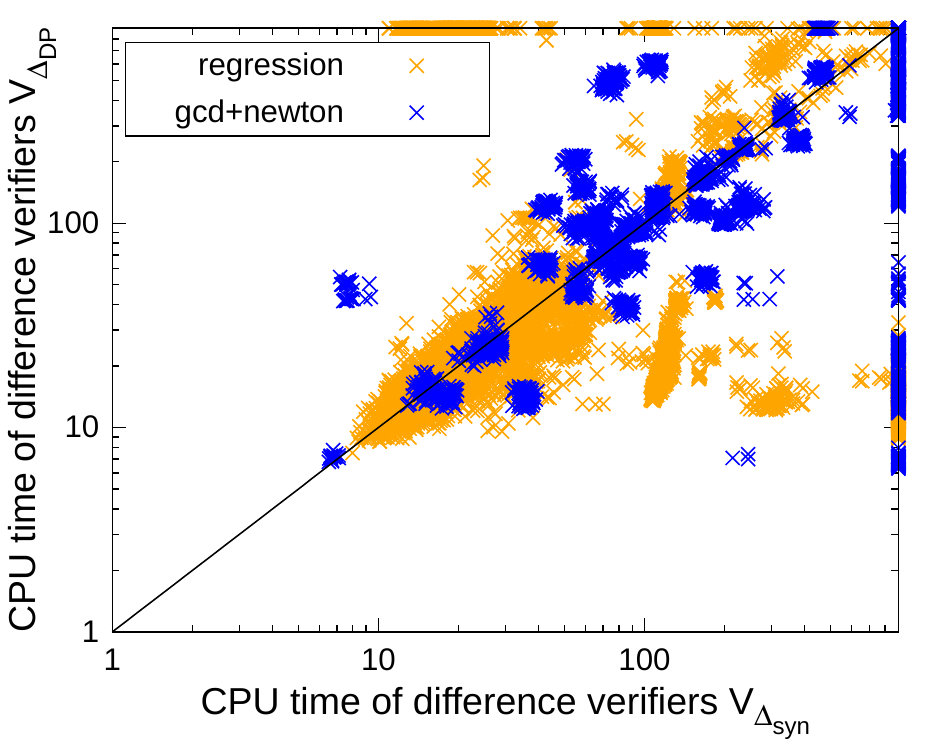}
		\hfill
		\includegraphics[width=0.3\textwidth]{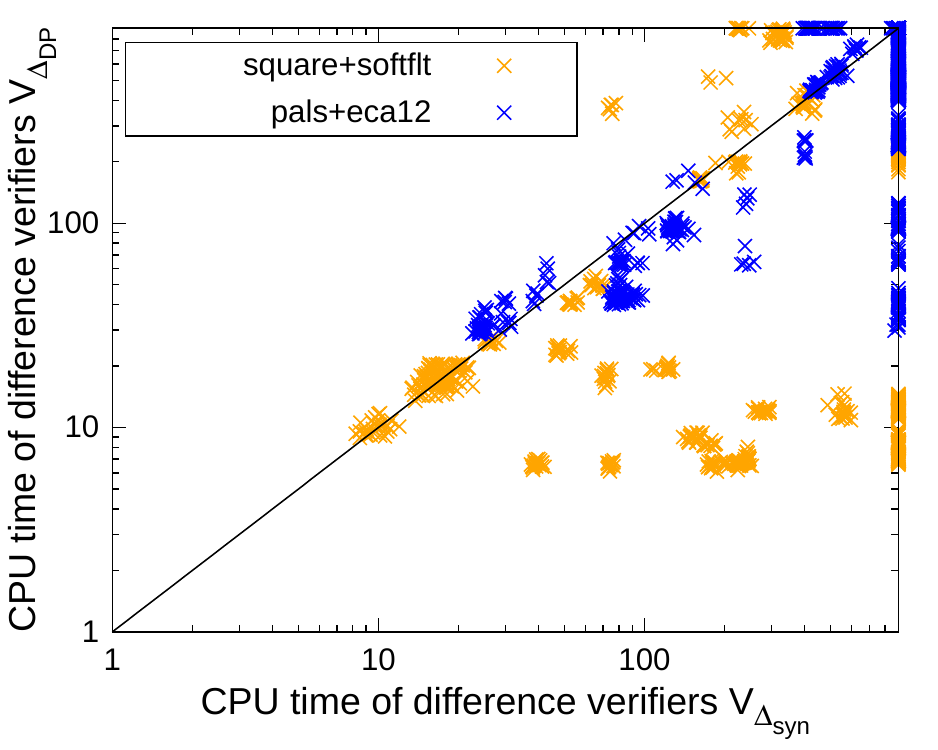}
		\vspace{-0.5em}
	\caption{Comparing CPU times of difference verifiers with syntax-based condition extractor (x-axis) and with our condition extractor (y-axis)}
	\label{fig:compareDiffV}
	\vspace{-1em}
\end{figure}

Next, we look at the efficiency in terms of CPU time. 
For each category (column in Tab.~\ref{tab:effectivityNumbers-comp-diffV}), we compare the efficiency of all pairs \((V^{\diffSynUn}, V^{\diffOursUn})\) of difference verifiers that use the same conditional verifier, but vary in the difference condition extractor, the scatter plots 
Figure~\ref{fig:compareDiffV} shows scatter plots that compare  the CPU times \(V^{\diffSynUn}\) (x-axis) and \(V^{\diffOursUn}\) (y-axis) needed to solve a task. 
Again, we use the maximum CPU time of \num{900}\,s if a task is not solved.
For \esbmc, we also exclude the results of the \num{10} incorrectly tasks.
Once again, the scatter plot does not consider CPU times of incorrectly solved tasks.
Studying the three scatter plots, we notice that for each category there exists data points in the lower right half, i.e., the difference verifier~\(V^{\diffOursUn}\) with our extractor is faster.
For categories \texttt{eca05+token} and \texttt{sfifo+token} (left scatter plot), the extra effort of our extractor rarely pays off.  
Computing a more precise condition does not pay off and often causes a slow down of the difference verifier because all paths with syntactical changes are already easy to verify. 
For categories \texttt{regression} and \texttt{gcd+newton} (middle scatter plot), the extra effort sometimes pay off. 
While for category \texttt{gcd+newton} in particular native \predicate and \esbmc profit from our more precise extractor, for \texttt{regression} all verifiers benefit for some tasks and for the categories~\texttt{square+softflt} and \texttt{pals+eca12} (right scatter plot) all verifiers often benefit from our extractor. 

\llbox{
We summarize that difference verification with our extractor is not always better than difference verification with the syntax-based difference condition extractor, but our extractor is more effective and efficient for several tasks.
}

\subsection{Threats to Validity}
\textbf{Threats to Internal Validity.}
Implementation bugs in one of the extractors, the residual program generator, or the verifiers may threaten the validity of our results.
We think that implementation bugs are unlikely and have at most little impact on our results.
On the one hand, we would expect that bugs in extractors or residual program generator result in (many) false results, in particular false proofs.
However, only \esbmc reported false results for ten tasks, for which another verifier succeeds.
On the other hand, the verifiers themselves participate in SV-COMP, therefore aiming to reduce the number of false results.
Furthermore, the verifiers may have benefited from the (structure of the) residual program instead of the condition extractors.
We think this is unlikely because we also observe positive effects when using a native conditional verifier, which does not use a residual program.
In addition, our experience from the past and this paper is that residual programs are often more difficult to verify.

\textbf{Threats to External Validity.}
Our results may not generalize to arbitrary programs.
Our combination tasks were artificially constructed by Beyer et al.~\cite{DiffCond} and are the result of the combination of two tasks, i.e., have a particular structure and differences may only occur in one of the combined tasks.
While the the regression tasks stem from real world tasks they only focus on the particular domain of Linux device drivers.
Also, our observations may not carry over to different verifiers and extractors.
Nevertheless, we confirmed the finding of Beyer et al.~\cite{DiffCond} for our new extractor.

\section{Related Work}
A broad range of approaches exist that aim to protect modified programs against regression bugs.
One line of work, deals with regression testing~\cite{RegressionTest}.
Other approaches (dis)prove behavior preservation, e.g.,~\cite{eXpress,BERT,PEQtest,RV-DAC,PEQcheck}, 
Also, there exist a variety of incremental verification approaches.
Like us, several of the incremental approaches skip the analysis of unchanged behavior~\cite{RegressionMC,DiSE,iDiSE,Memoise,LeinoW15,iCoq,RPS,DiffCond}.
Some incremental verification approaches~adapt previous analysis results including fixed points~\cite{Reviser,IncA,DBLP:conf/scam/PlasSER20} or the explored state space~\cite{ExtremeMC,ISSE,EvolCheck} while others reuse intermediate results like constraint solutions~\cite{greenModelChecking,Recal}, abstraction details~\cite{PrecisionReuse,TraceAbstractionReuse,P2V}, or annotations~\cite{SummaryReuse,AnnotationReuse}.

Next, let us have a closer look of what unchanged parts of a program are skipped in the different incremental verification approaches.
iCoq~\cite{iCoq} addresses proofs in the theorem prover Coq and only reverifies Coq~proofs that are affected by the modification.
The deductive verifier Dafny\cite{LeinoW15} only reverifies methods affected by the modification and reuses unaffected verification conditions obtained in previous verification runs.
In contrast, regression property selection~(RPS)~\cite{RPS} restricts runtime verification to properties that are related to classes that are affected by the modification.
Furthermore, DiSE~\cite{DiSE,iDiSE} is an incremental technique for symbolic execution that directs symbolic execution towards program statements that are affected by the modification.
There also exist approaches that stop exploration as soon as a path cannot reach a change.
For instance, RMC~\cite{RegressionMC} stops DFS exploration in this case, while Memoise~\cite{Memoise} stops symbolic execution of states that cannot reach any modification and disables constraint solving on unchanged paths.
Originally, difference verification with condition~\cite{DiffCond} stops if a path cannot reach a syntactic difference.
We extend the framework of difference verification with conditions with a new difference condition extractor that stops if no modification can be reached or we can show based on data- and control-flow that the modifications seen cannot affect the property of interest, and, thus, cannot cause a regression bug.

After we compared what is skipped, we now look into which information is used to determine what is affected.
Dafny~\cite{LeinoW15} relies on checksums and iCoq~\cite{iCoq} relies on timestamps and checksums to determine affected components. 
In contrast, the syntax-based difference condition extractor from~\cite{DiffCond}, Memoise~\cite{Memoise}, and RMC~\cite{RegressionMC} treat all syntactically changed paths as affected. 
Going beyond pure syntactic differences, RPS~\cite{RPS} uses class dependencies to compute affected classes.
In contrast, our difference condition extractor mainly relies on data dependencies.
Similarly, DiSE~\cite{DiSE,iDiSE} considers control and data dependencies to detect change affected program statements and Jana et al.~\cite{DBLP:conf/issre/JanaKCKG021} use them to detect which properties are affected by the modification at which locations.
Finally, we want to state that there exist approaches that use abstract interpretation~\cite{CondEQ,Dizy} or symbolic execution~\cite{DiffSymExe,PartitionBasedRegressionVerification,ModDiff,ShadowSymbolicExecution} to determine the inputs and, thus, the paths for which the behavior of original and modified programs differ functionally.
So far, they have not been used to skip unaffected paths in reverification.
\section{Conclusion}
Every software modification may introduce new bugs and, thus, threaten software reliability.
When software verification is used to establish software reliability, the software verifiers need to keep up with the frequent software changes.
Verifying the complete modified program after each change is typically too costly.
Therefore, we require incremental verifiers that rely on the insight that changes do not introduce regression bugs everywhere in the software.

One incremental and also verifier agnostic approach is difference verification with conditions.
It uses the idea of conditional model checking to restrict the verification to those program parts which may contain new bugs introduced by the change.
To this end, it precedes the conditional model checker with a difference condition extractor that performs a change analysis to detect the program paths relevant for reverification.
Like the verifier, also the difference condition extractor can be exchanged.
While the proof-of-concept implementation, uses a syntax-based difference condition extractor, we suggest a more complex difference condition extractor that also considers data dependencies and program properties during extraction.
We prove its soundness and experimentally show for various verifiers that this new difference detector improves the effectiveness and efficiency of difference verification with condition on some tasks.


%
 \bibliographystyle{splncs04}
 \bibliography{literature}
\clearpage
\appendix
\section{Proving Soundness of \diffDetectOurs}
In the following, we write \(paths(P)\) for the set of executable program paths of program~\(P=(L,\ell_0,G, \ell_\mathrm{err})\) and \(c=_{|_{V^-}}c'\) to denote that concrete data states \(c, c'\) map all variables except for the variables from the set \(V\) to the same value.

\subsection{Helper Lemmas for Proving Theorem~2}
We start with a set of lemmas that simplify the proof of Theorem~2.
First, we show that each prefix of any path with a regression bug, which is a path of the modified program, can be related to at most one path in the difference graph produced by our difference detector~\diffDetectOurs. 

\begin{lemma}\label{lem:deterministic}
Let \(P=(L,\ell_0,G, \ell_\mathrm{err})\) and \(P'=(L',\ell'_0,G', \ell'_\mathrm{err})\) be two deterministic programs.
For all \(\pi' = (\ell'_0,c'_0) \stackrel{g'_1}{\rightarrow}(\ell'_1, c'_1) \stackrel{g'_2}{\rightarrow} \ldots \stackrel{g'_n}{\rightarrow} (\ell'_n, c'_n)\in paths^\mathrm{rb}(P,P')\) and \(0{\leq}k{\leq}n\) if \(n_0\stackrel{g'_1}{\rightarrow}n_1 \stackrel{g'_2}{\rightarrow} \ldots \stackrel{g'_k}{\rightarrow} n_k\in \textnormal{\diffDetectOurs}(P,P')\) and \(n'_0\stackrel{g'_1}{\rightarrow}n'_1 \stackrel{g'_2}{\rightarrow} \ldots \stackrel{g'_k}{\rightarrow} n'_k\in \textnormal{\diffDetectOurs}(P,P')\) such that \(n'_0=n_0\), then \(\forall 0{\leq}i{\leq}k: n_i=n'_i\).
\end{lemma}

\begin{proof}
Let \(P=(L,\ell_0,G, \ell_\mathrm{err})\) and \(P'=(L',\ell'_0,G', \ell'_\mathrm{err})\) be arbitrary deterministic programs, \(\pi' = (\ell'_0,c'_0) \stackrel{g'_1}{\rightarrow}(\ell'_1, c'_1) \stackrel{g'_2}{\rightarrow} \ldots \stackrel{g'_n}{\rightarrow} (\ell'_n, c'_n)\in paths^\mathrm{rb}(P,P')\) be arbitrary and \diffDetectOurs\((P,P')=(N,E,n_0,\bad)\).
Choose arbitrary \(k\in\mathbb{N}\) with \(0{\leq}k{\leq}n\) and let
\(n_0\stackrel{g'_1}{\rightarrow}n_1 \stackrel{g'_2}{\rightarrow} \ldots \stackrel{g'_k}{\rightarrow} n_k\in \textnormal{\diffDetectOurs}(P,P')\) and \( n'_0\stackrel{g'_1}{\rightarrow}n'_1 \stackrel{g'_2}{\rightarrow} \ldots \stackrel{g'_k}{\rightarrow} n'_k\in \textnormal{\diffDetectOurs}(P,P')\) be arbitrary, but both starting at the root node, i.e., \(n'_0=n_0\).
Proof by induction that \(\forall 0{\leq}i{\leq}k: n_i=n'_i\).
\begin{description}
  \item[Base case ($i=0$)] Follows from assumption.
  \item[Step case ($i-1\rightarrow i, i\leq k$)]
	By induction hypothesis,  \(\forall 0{\leq}j{\leq}i-1: n_j=n'_j\).
	We need to show that \(n_i=n'_i\).
  By construction \(N=\visited\cup\bad\), for each \(\ell'\in L'\) not exists \((\cdot,\ell')\in\visited\) or  \(\ell'\notin\bad\), and \(E\subseteq \visited\times G'\times (\visited\cup\bad)\).
  Also, \(\forall 1{\leq}i{\leq}k: g'_i=(\ell'_{i-1},\cdot, \ell'_i)\Rightarrow n_{i-1}=(\cdot,\ell'_{i-1})\wedge n'_{i-1}=(\cdot,\ell'_{i-1}) \wedge (n_{i}=(\cdot,\ell'_{i})\wedge n'_{i}=(\cdot,\ell'_{i}) \vee n_{i}=\ell'_{i}=n'_{i}\wedge i=k)\).
	Let \(g'_i=(\ell'_{i-1},op'_i, \ell'_i)\).
	Hence, \(n_i=\ell'_i=n'_i\) or there exists \(\ell, \ell_1,\ell_2\in L\) such that \(n_{i-1}=n'_{i-1}=(\ell, \ell'_{i-1})\), \(n_i=(\ell_1, \ell'_{i})\), \(n'_{i}=(\ell_2, \ell'_{i})\), \((n_{i-1}, (\ell'_{i-1}, op'_i, \ell'_i), (\ell_1, \ell'_i))\in E\), and \((n'_{i-1}, (\ell'_{i-1}, op'_i, \ell'_i), (\ell_2, \ell'_i))\in E\).
	Since \(P, P'\) are deterministic, all computed edges \(((\ell_p, \ell'_p),g',(\ell_s,\ell'_s))\) that Alg.~\ref{alg:diffcompprop} adds to \(\edges\) are deterministic.
	Due to construction of \(E\), we know that \(((\ell, \ell'_{i-1}), (\ell'_{i-1}, op_i, \ell'_i), (\ell_1, \ell'_i))\in E\) and  \(((\ell', \ell'_{i-1}), (\ell'_{i-1}, op_i, \ell'_i), (\ell_2, \ell'_i))\in E\) if \(((\ell, \ell'_{i-1}), (\ell'_{i-1}, op_i, \ell'_i), (\ell_1, \ell'_i))\in \edges\), \(((\ell', \ell'_{i-1}), (\ell'_{i-1}, op_i, \ell'_i), (\ell_2, \ell'_i))\in\edges\).
	We conclude \(\ell_1=\ell_2\).
	Thus, \(n_i=n'_i\).
 \end{description}
\qed
\end{proof}

Next, we show that the difference detector \diffDetectOurs properly extends any alignment constructed for a real prefix of a path with regression bug if it does not end in a regression bug indicator node.

\begin{lemma}\label{lem:extend}
Let \(P=(L,\ell_0,G, \ell_\mathrm{err})\) and \(P'=(L',\ell'_0,G', \ell'_\mathrm{err})\) be arbitrary deterministic programs, \diffDetectOurs\((P,P')=(N,E,n_0,\bad)\), and \(\modified\) in the state as at the end of \diffDetectOurs\((P,P')\). 
For all \(\pi' = (\ell'_0,c'_0) \stackrel{g'_1}{\rightarrow}(\ell'_1, c'_1) \stackrel{g'_2}{\rightarrow} \ldots \stackrel{g'_n}{\rightarrow} (\ell'_n, c'_n)\in paths^\mathrm{rb}(P,P')\) and \(0{\leq}k{<}n\) if \(n_0\stackrel{g'_1}{\rightarrow}n_1 \stackrel{g'_2}{\rightarrow} \ldots \stackrel{g'_k}{\rightarrow} n_k\in \textnormal{\diffDetectOurs}(P,P')\) and there exists \((\ell_0,c_0) \stackrel{g_1}{\rightarrow}(\ell_1, c_1) \stackrel{g_2}{\rightarrow} \ldots \stackrel{g_m}{\rightarrow} (\ell_m, c_m)\in paths(P)\) with \(c_0=c'_0\), and there exists total, monotonously increasing function \(f:\{0,\dots,k\}\rightarrow\{0,\dots,m\}\) with \(f(k)=m\) and \(\forall 0{\leq}i{\leq}k: n_i=(\ell_{f(i)},\ell'_i)\wedge c'_i=_{|_{\modified(n_i)^-}}c_{f(i)}\), then there exists \((n_k, g'_{k+1},n_{k+1})\in E\) with \(n_{k+1}\in\Delta\) or there exist \((\ell_0,c_0) \stackrel{g_1}{\rightarrow}(\ell_1, c_1) \stackrel{g_2}{\rightarrow} \ldots \stackrel{g_m}{\rightarrow} (\ell_m, c_m)\ldots \stackrel{g_{m+t}}{\rightarrow} (\ell_{m+t}, c_{m+t})\in paths(P)\) and total, monotonously increasing function \(f':\{0,\dots,k+1\}\rightarrow\{0,\dots,m+t\}\) with \(f'(k+1)={m+t}\) and \(\forall 0{\leq}i{\leq}k+1: n_i=(\ell_{f'(i)},\ell'_i)\wedge c'_i=_{|_{\modified(n_i)^-}}c_{f'(i)}\).
\end{lemma}

\begin{proof}
Let \(P=(L,\ell_0,G, \ell_\mathrm{err})\) and \(P'=(L',\ell'_0,G', \ell'_\mathrm{err})\) be arbitrary deterministic programs, \diffDetectOurs\((P,P')=(N,E,n_0,\bad)\), and \(\modified\) in the state as at the end of \diffDetectOurs\((P,P')\).
Let \(\pi' = (\ell'_0,c'_0) \stackrel{g'_1}{\rightarrow}(\ell'_1, c'_1) \stackrel{g'_2}{\rightarrow} \ldots \stackrel{g'_n}{\rightarrow} (\ell'_n, c'_n)\in paths^\mathrm{rb}(P,P')\), \(k\in\mathbb{N}\) such that \(0{\leq}k{<}n\) and \(n_0\stackrel{g'_1}{\rightarrow}n_1 \stackrel{g'_2}{\rightarrow} \ldots \stackrel{g'_k}{\rightarrow} n_k\in \textnormal{\diffDetectOurs}(P,P')\) be arbitrary such that there exists \(\pi=(\ell_0,c_0) \stackrel{g_1}{\rightarrow}(\ell_1, c_1) \stackrel{g_2}{\rightarrow} \ldots \stackrel{g_m}{\rightarrow} (\ell_m, c_m)\in paths(P)\) with \(c_0=c'_0\) and there exists total,  monotonously increasing function~\(f:\{0,\dots,k\}\rightarrow\{0,\dots,m\}\) with \(f(k)=m\) and \(\forall 0{\leq}i{\leq}k: n_i=(\ell_{f(i)},\ell'_i)\wedge c'_i=_{|_{\modified(n_i)^-}}c_{f(i)}\).
In the following, we need to show that there exists \((n_k, g'_{k+1},n_{k+1})\in E\) s.t.\ \(n_{k+1}\in\Delta\) or there exists \(\pi_n=(\ell_0,c_0) \stackrel{g_1}{\rightarrow}(\ell_1, c_1) \stackrel{g_2}{\rightarrow} \ldots \stackrel{g_m}{\rightarrow} (\ell_m, c_m)\ldots \stackrel{g_{m+t}}{\rightarrow} (\ell_{m+t}, c_{m+t})\in paths(P)\) and total, monotonously increasing function~\(f':\{0,\dots,k+1\}\rightarrow\{0,\dots,m+t\}\) with \(f'(k+1)={m+t}\) and \(\forall 0{\leq}i{\leq}k+1: n_i=(\ell_{f'(i)},\ell'_i)\wedge c'_i=_{|_{\modified(n_i)^-}}c_{f'(i)}\).
Since \(\pi'\in paths^\mathrm{rb}(P,P')\), \(\pi\in paths(P)\), and \(c_0=c'_0\), we know \(\forall 0\leq i\leq m: \ell_i\neq\ell_\mathrm{err}\).
Hence, by construction for all \( 0\leq i\leq k, n_i\in (L\setminus\{\ell_\mathrm{err}\})\times (L'\setminus(\{\ell'_\mathrm{err}\}\cup\Delta))\) and Alg.~\ref{alg:diffcompprop} ensures that its while loop explores \(n_k\) with \(\modified(n_k)\).
Furthermore, \(n_k=(\ell_{f(k)},\ell_k)=(\ell_m, \ell_k)\). 
Due to semantics and \(k<n\), there exists \(g'_{k+1}=(\ell'_k, op'_{k+1}, \ell'_{k+1})\) for some \(op'_{k+1}\in Ops\).
We distinguish between the operation type of \(op'_{k+1}\).

\begin{description}
  \item[Case 1: $op'_{k+1}$ assignment]
Due to the semantics, \(c'_k=_{|_{\writ{(op'_{k+1})}^-}}c'_{k+1}\).
Now, consider two cases.

\begin{enumerate}
  \item Assume there exists \(\ell_{m+1}\in L\) with \(g_{k+1}=(\ell_m, op'_{k+1},\ell_{m+1})\in G\).
Due to semantics, there exists transition step \((\ell_m, c_m)\stackrel{g_{k+1}}{\rightarrow} (\ell_{m+1}, c_{m+1})\) and \(c_m=_{|_{\writ({op'_{k+1}})^-}}c_{m+1}\).
From \(c'_k=_{|_{{\modified((\ell_m,\ell'_k))}^-}}c_{m}\), we conclude that \(c'_{k+1}=_{|_{{(\modified((\ell_m,\ell'_k))\cup\writ({op'_{k+1}}))}^-}}c_{m+1}\).
Furthermore, we conclude that if \(\modified((\ell_m,\ell'_k))\cap\rd({op'_{k+1}})=\emptyset\), from the semantics it then follows that \(c'_{k+1}=_{|_{{(\modified((\ell_m,\ell'_k))\setminus\writ({op'_{k+1}}))}^-}}c_{m+1}\).
Let $V_\mathrm{diff}=\modified((\ell_m,\ell'_k))\setminus\writ(op'_{k+1})$ if \(\modified((\ell_m,\ell'_k))\cap\rd({op'_{k+1}})=\emptyset\) and $V_\mathrm{diff}=\modified((\ell_m,\ell'_k))\cup\writ(op'_{k+1})$ otherwise.
Then, we can derive \(c'_{k+1}=_{|_{^{{V_\mathrm{diff}}^-}}}c_{m+1}\).
Also, Alg.~\ref{alg:diffcompprop} either  (a)~adds\(((\ell_m, \ell'_k),g'_{k+1},\ell'_{k+1})\) to \(\edges\) and adds \(\ell'_{k+1}\) to \(\bad\), (b)~adds \(((\ell_m, \ell'_k), g'_{k+1},(\ell_{m+1}, \ell'_{k+1}))\) to \(\edges\) and already added \(\ell'_{k+1}\) to \(\bad\), or (c)~it adds \(((\ell_m, \ell'_k), g'_{k+1},(\ell_{m+1}, \ell'_{k+1}))\) to \(\edges\), \((\ell_{m+1}, \ell'_{k+1})\) to \(\visited\), and \(V_\mathrm{diff}\) to \(\modified((\ell_{m+1},\ell'_{k+1}))\).
Due to construction of \(E\) and never deleting from \(\bad,\edges, \modified\), there either exists \(((\ell_m, \ell'_k),g'_{k+1},\ell'_{k+1})\in E\) with \(\ell'_{k+1}\in\bad\) or \(((\ell_m, \ell'_k), g'_{k+1},(\ell_{m+1}, \ell'_{k+1}))\in E\) with \(V_\mathrm{diff}\subseteq \modified((\ell_{m+1},\ell'_{k+1}))\).
In the latter case, we set \(f'=f\cup\{(k+1, m+1)\}\) and \(\pi_n=(\ell_0,c_0) \stackrel{g_1}{\rightarrow}(\ell_1, c_1) \stackrel{g_2}{\rightarrow} \ldots \stackrel{g_m}{\rightarrow} (\ell_m, c_m)\stackrel{g_{m+1}}{\rightarrow} (\ell_{m+1}, c_{m+1})\).
Since \(V_\mathrm{diff}\subseteq \modified((\ell_{m+1},\ell'_{k+1}))\), also \(c'_{k+1}=_{|_{^{{\modified((\ell_{m+1},\ell'_{k+1}))}^-}}}c_{m+1}\).
From \(f\) being total and monotonously increasing with largest value \(m\), we conclude \(f'\) is total and monotonously increasing.
By construction of \(f'\), we get \(\forall 0{\leq}i{\leq}k+1: n_i=(\ell_{f'(i)},\ell'_i)\wedge c'_i=_{|_{\modified(n_i)^-}}c_{f'(i)}\).

\item Assume there does not exit \(\ell_{m+1}\in L\) s.t.\ \(g_{k+1}=(\ell_m, op'_{k+1},\ell_{m+1})\in G\).
From \(c'_k=_{|_{\modified((\ell_m,\ell'_k))^-}}c_{m}\) and  \(c'_k=_{|_{\writ{(op'_{k+1})}^-}}c'_{k+1}\), we get that \(c'_{k+1}=_{|_{^(\modified((\ell_m,\ell'_k))\cup\writ({op'_{k+1}}))^-}}c_{m}\).
Also, Alg.~\ref{alg:diffcompprop} either  (a)~adds \(\ell'_{k+1}\) to \(\bad\) and\(((\ell_m, \ell'_k),g'_{k+1},\ell'_{k+1})\) to \(\edges\), (b) adds \(((\ell_m, \ell'_k), g'_{k+1},(\ell_{m}, \ell'_{k+1}))\) to \(\edges\) and already added \(\ell'_{k+1}\) to \(\bad\), or (c)~it adds \(((\ell_m, \ell'_k), g'_{k+1},(\ell_{m}, \ell'_{k+1}))\) to \(\edges\), \((\ell_{m}, \ell'_{k+1})\) to \(\visited\), and \(V_\mathrm{diff}\) to \(\modified((\ell_{m},\ell'_{k+1}))\).
Due to construction of \(E\) and never deleting elements from \(\bad,\edges, \modified\),  there either exists \(((\ell_m, \ell'_k),g'_{k+1},\ell'_{k+1})\in E\) with \(\ell'_{k+1}\in\bad\) or \(((\ell_m, \ell'_k), g'_{k+1},(\ell_m, \ell'_{k+1}))\in E\) with \((\modified((\ell_m,\ell'_k))\cup\writ({op'_{k+1}}))\subseteq \modified((\ell_m,\ell'_{k+1}))\).
In the latter case, set \(\pi_n=\pi\) and \(f'=f\cup\{(k+1, m)\}\).
Then, we conclude from  \((\modified((\ell_m,\ell'_k))\cup\writ({op'_{k+1}}))\subseteq \modified((\ell_m,\ell'_{k+1}))\) that also \(c'_{k+1}=_{|_{^(\modified((\ell_m,\ell'_{k+1}))\cup\writ({op'_{k+1}}))^-}}c_{m}\).
From \(f\) being total and monotonously increasing with largest value \(m\), we infer \(f'\) is total and monotonously increasing.
By construction of \(f'\), we get \(\forall 0{\leq}i{\leq}k+1: n_i=(\ell_{f'(i)},\ell'_i)\wedge c'_i=_{|_{\modified(n_i)^-}}c_{f'(i)}\).
\end{enumerate}

\item[Case 2: $op'_{k+1}$ assume operation]
Let \(r\) be the number of iterations of the while loop in lines~13--14, in which that while loop explores assignments \((\ell^p_{m+1},op_{m+1}, \ell_{m+1})\),\dots,\((\ell^p_{m+r},op_{m+r}, \ell_{m+r})\) and computes change affected variables~\(V_\mathrm{diff}=\modified((\ell_m,\ell'_k))\bigcup_{i=1}^{r} \writ(op_{m+i})\).
By construction, we get \(\ell^p_{m+1}=\ell_m\) and \(\forall 2\leq i\leq r: \ell_{m+i-1}=\ell^p_{m+i}\).
Due to semantics, there exists \((\ell_m,c_m) \stackrel{g_{m+1}}{\rightarrow}(\ell_{m+1}, c_{m+1}) \stackrel{g_{m+2}}{\rightarrow} \ldots \stackrel{g_{m+r}}{\rightarrow}(\ell_{m+r}, c_{m+r})\) s.t.\ for all \mbox{\(1\leq i\leq r\)}, \(c_{m+i-1}=_{|_{{\bigcup_{j=1}^{i} \writ(op_{k+j})}^-}} c_{m+i}\).
Furthermore, \(\pi_i=(\ell_0,c_0) \stackrel{g_1}{\rightarrow}(\ell_1, c_1) \stackrel{g_2}{\rightarrow} \ldots \stackrel{g_m}{\rightarrow} (\ell_m, c_m)\stackrel{g_{m+1}}{\rightarrow}(\ell_{m+1}, c_{m+1}) \stackrel{g_{m+2}}{\rightarrow} \ldots \stackrel{g_{m+r}}{\rightarrow}(\ell_{m+r}, c_{m+r})\in paths(P)\) and after the while loop \(\ell_p=\ell_{m+r}\).
From \(c'_k=_{|_{{\modified((\ell_m,\ell'_k))}^-}}c_{m}\), we conclude \(c'_k=_{|_{{V_\mathrm{diff}}^-}}c_{m+r}\).
Due to semantics, \(c'_{k+1}=c'_k\).
Hence, \(c'_{k+1}=_{|_{{V_\mathrm{diff}}^-}}c_{m+r}\).
Now consider four cases.

\begin{enumerate}
	\item  Assume \(\ell_p\neq\ell_\mathrm{err}\), there exists \((\ell_p,op'_{k+1}, \ell_{m+r+1})\in G\), and \(\rd(op'_{k+1})\cap V_\mathrm{diff}=\emptyset\).
Since \((\ell'_k,c'_k)\stackrel{g'_{k+1}}{\rightarrow}(\ell'_{k+1}, c'_{k+1})\) implies \(c'_{k}\models op'_{k+1}\) and we know \(c'_k=_{|_{{V_\mathrm{diff}}^-}}c_{m+r}\) and \(\rd(op'_{k+1})\cap V_\mathrm{diff}=\emptyset\), also \(c_{m+r}\models op'_{k+1}\).
Due to semantics and \(\ell_p=\ell_{m+r}\), there exists \((\ell_{m+r}, c_{m+r})\xrightarrow{(\ell_{m+r},op'_{k+1}, \ell_{m+r+1})}(\ell_{m+r+1},c_{m+r+1})\) with \(c_{m+r+1}=c_{m+r}\) and, thus, \(c'_{k+1}=_{|_{{V_\mathrm{diff}}^-}}c_{m+r+1}\).
Also, \(\pi_n=(\ell_0,c_0) \stackrel{g_1}{\rightarrow}(\ell_1, c_1) \stackrel{g_2}{\rightarrow} \ldots \stackrel{g_m}{\rightarrow} (\ell_m, c_m)\stackrel{g_{m+1}}{\rightarrow}(\ell_{m+1}, c_{m+1}) \stackrel{g_{m+2}}{\rightarrow} \ldots  \stackrel{g_{m+r}}{\rightarrow}(\ell_{m+r}, c_{m+r})\xrightarrow{(\ell_{m+r},op'_{k+1}, \ell_{m+r+1})}(\ell_{m+r+1},c_{m+r+1})\in paths(P)\).
Algorithm~\ref{alg:diffcompprop} either  (a)~adds \(\ell'_{k+1}\) to \(\bad\) and\(((\ell_m, \ell'_k),g'_{k+1},\ell'_{k+1})\) to \(\edges\), 
(b) adds \(((\ell_m, \ell'_k),g'_{k+1},(\ell_{m+r+1}, \ell'_{k+1})\) to \(\edges\) and already added \(\ell'_{k+1}\) to \(\bad\), or (c) adds \(((\ell_m, \ell'_k), g'_{k+1},(\ell_{m+r+1}, \ell'_{k+1}))\) to \(\edges\), \((\ell_{m+r+1}, \ell'_{k+1})\) to \(\visited\), and \(V_\mathrm{diff}\) to \(\modified((\ell_{m+r+1},\ell'_{k+1}))\).
Due to construction of \(E\) and never deleting elements from \(\bad,\edges, \modified\), there either exists 
\(((\ell_m, \ell'_k),g'_{k+1},\ell'_{k+1})\in E\) with \(\ell'_{k+1}\in\bad\) or \(((\ell_m, \ell'_k), g'_{k+1},(\ell_{m+r+1}, \ell'_{k+1}))\in E\) with \(V_\mathrm{diff}\subseteq \modified((\ell_{m+r+1},\ell'_{k+1}))\).
In the latter case, consider \(\pi_n\) and set \(f'=f\cup\{(k+1, m+r+1)\}\).
Since \(V_\mathrm{diff}\subseteq \modified((\ell_{m+r+1},\ell'_{k+1}))\) and \(c'_{k+1}=_{|_{{V_\mathrm{diff}}^-}}c_{m+r+1}\), also \(c'_{k+1}=_{|_{\modified((\ell_{m+r+1},\ell'_{k+1}))^-}}c_{m+r+1}\).
From \(f\) being total and monotonously increasing with largest value \(m\), we infer \(f'\) is total and monotonously increasing.
By construction of \(f'\), we get \(\forall 0{\leq}i{\leq}k+1: n_i=(\ell_{f'(i)},\ell'_i)\wedge c'_i=_{|_{\modified(n_i)^-}}c_{f'(i)}\).
  \item Assume \(\ell_p\neq\ell_\mathrm{err}\), there exists \((\ell_p,op'_{m+r+1}, \ell_{m+r+1})\in G\) s.t.\ \(op_{m+r+1}\) is satisfiable, \(op'_{k+1}\implies op_{m+r+1}\), and \((\rd(op'_{k+1})\cup\rd(op_{m+r+1}))\cap V_\mathrm{diff}=\emptyset\).
Since \((\ell'_k,c'_k)\stackrel{g'_{k+1}}{\rightarrow}(\ell'_{k+1}, c'_{k+1})\) implies \(c'_{k}\models op'_{k+1}\) and we know \(c'_k=_{|_{{V_\mathrm{diff}}^-}}c_{m+r}\) and \(\rd(op'_{k+1})\cap V_\mathrm{diff}=\emptyset\), also \(c_{m+r}\models op'_{k+1}\).
Due to \(op'_{k+1}\implies op_{m+r+1}\), furthermore \(c_{m+r}\models op_{m+r+1}\).
Due to semantics and \(\ell_p=\ell_{m+r}\), there exists \((\ell_{m+r}, c_{m+r})\xrightarrow{(\ell_{m+r},op_{m+r+1}, \ell_{m+r+1})}(\ell_{m+r+1},c_{m+r+1})\) with \(c_{m+r+1}=c_{m+r}\) and, thus, \(c'_{k+1}=_{|_{{V_\mathrm{diff}}^-}}c_{m+r+1}\).
Also, \(\pi_n=(\ell_0,c_0) \stackrel{g_1}{\rightarrow}(\ell_1, c_1) \stackrel{g_2}{\rightarrow} \ldots \stackrel{g_m}{\rightarrow} (\ell_m, c_m)\stackrel{g_{m+1}}{\rightarrow}(\ell_{m+1}, c_{m+1}) \stackrel{g_{m+2}}{\rightarrow} \ldots  \stackrel{g_{m+r}}{\rightarrow}(\ell_{m+r}, c_{m+r})\xrightarrow{(\ell_{m+r},op_{m+r+1}, \ell_{m+r+1})}(\ell_{m+r+1},c_{m+r+1})\in paths(P)\).
Algorithm~\ref{alg:diffcompprop} either  (a)~adds \(\ell'_{k+1}\) to \(\bad\) and\(((\ell_m, \ell'_k),g'_{k+1},\ell'_{k+1})\) to \(\edges\), 
(b) adds \(((\ell_m, \ell'_k),g'_{k+1},(\ell_{m+r+1}, \ell'_{k+1})\) to \(\edges\) and already added \(\ell'_{k+1}\) to \(\bad\), or (c) adds \(((\ell_m, \ell'_k), g'_{k+1},(\ell_{m+r+1}, \ell'_{k+1}))\) to \(\edges\), \((\ell_{m+r+1}, \ell'_{k+1})\) to \(\visited\), and \(V_\mathrm{diff}\) to \(\modified((\ell_{m+r+1},\ell'_{k+1}))\).
Due to construction of \(E\) and never deleting elements from \(\bad,\edges, \modified\), there either exists 
\(((\ell_m, \ell'_k),g'_{k+1},\ell'_{k+1})\in E\) with \(\ell'_{k+1}\in\bad\) or \(((\ell_m, \ell'_k), g'_{k+1},(\ell_{m+r+1}, \ell'_{k+1}))\in E\) with \(V_\mathrm{diff}\subseteq \modified((\ell_{m+r+1},\ell'_{k+1}))\).
In the latter case, consider \(\pi_n\) and set \(f'=f\cup\{(k+1, m+r+1)\}\).
Since \(V_\mathrm{diff}\subseteq \modified((\ell_{m+r+1},\ell'_{k+1}))\) and \(c'_{k+1}=_{|_{{V_\mathrm{diff}}^-}}c_{m+r+1}\), also \(c'_{k+1}=_{|_{\modified((\ell_{m+r+1},\ell'_{k+1}))^-}}c_{m+r+1}\).
From \(f\) being total and monotonously increasing with largest value \(m\), we infer \(f'\) is total and monotonously increasing.
By construction of \(f'\), we get \(\forall 0{\leq}i{\leq}k+1: n_i=(\ell_{f'(i)},\ell'_i)\wedge c'_i=_{|_{\modified(n_i)^-}}c_{f'(i)}\).

  \item Assume \(\ell_p\neq\ell_\mathrm{err}\), there exists  \((\ell_p,op'_{k+1}, \ell_2)\in G\) s.t.\ \(\rd(op'_{k+1})\cap V_\mathrm{diff}\neq\emptyset\) or there exists \((\ell_p,op_{m+r+1}, \ell_2)\in G\) with \(op_{m+r+1}\) being satisfiable, \(op'_{k+1}\implies op_{m+r+1}\), and \((\rd(op'_{k+1})\cup\rd(op_{m+r+1}))\cap V_\mathrm{diff}\neq\emptyset\).
 Algorithm~\ref{alg:diffcompprop} adds \(((\ell_m, \ell'_k),g'_{k+1},\ell'_{k+1})\) to \(\edges\) and adds or already added \(\ell'_{k+1}\) to \(\bad\).
Due to construction of \(E\), there exists
\(((\ell_m, \ell'_k),g'_{k+1},\ell'_{k+1})\in E\) with \(\ell'_{k+1}\in\bad\).	
	
  \item Assume \(\ell_p=\ell_\mathrm{err}\) or there does not exist \((\ell_p,op'_{k+1}, \ell_2)\in G\) and there does not exist \((\ell_p,op_{m+r+1}, \ell_2)\in G\) with \(op_{m+r+1}\) being satisfiable and \(op'_{k+1}\implies op_{m+r+1}\).
 Since \(e_p=\ell_{m+r}\), Alg.~\ref{alg:diffcompprop} either  (a)~adds \(\ell'_{k+1}\) to \(\bad\) and\(((\ell_m, \ell'_k),g'_{k+1},\ell'_{k+1})\) to \(\edges\), (b) adds \(((\ell_m, \ell'_k), g'_{k+1},(\ell_{m+r}, \ell'_{k+1}))\) to \(\edges\) and already added \(\ell'_{k+1}\) to \(\bad\), or (c)~it adds \(((\ell_m, \ell'_k), g'_{k+1},(\ell_{m+r}, \ell'_{k+1}))\) to \(\edges\), \((\ell_{m+r}, \ell'_{k+1})\) to \(\visited\), and \(V_\mathrm{diff}\) to \(\modified((\ell_{m+r},\ell'_{k+1}))\).
Due to construction of \(E\) and never deleting elements from \(\bad,\edges, \modified\), there either exists 
\(((\ell_m, \ell'_k),g'_{k+1},\ell'_{k+1})\in E\) with \(\ell'_{k+1}\in\bad\) or there exists \(((\ell_m, \ell'_k), g'_{k+1},(\ell_{m+r}, \ell'_{k+1}))\in E\) with \(V_\mathrm{diff}\subseteq \modified((\ell_{m+r},\ell'_{k+1}))\).
In the latter case, set \(\pi_n=\pi_i\) and \(f'=f\cup\{(k+1, m+r)\}\).
Since \(V_\mathrm{diff}\subseteq \modified((\ell_{m+r},\ell'_{k+1}))\) and \(c'_{k+1}=_{|_{{V_\mathrm{diff}}^-}}c_{m+r}\), also \(c'_{k+1}=_{|_{^\modified((\ell_{m+r},\ell'_{k+1}))^-}}c_{m+r}\).
From \(f\) being total and monotonously increasing with largest value \(m\), we infer \(f'\) is total and monotonously increasing.
By construction of \(f'\), we get \(\forall 0{\leq}i{\leq}k+1: n_i=(\ell_{f'(i)},\ell'_i)\wedge c'_i=_{|_{\modified(n_i)^-}}c_{f'(i)}\).
\end{enumerate}
\end{description}
\qed
\end{proof}

Based on the previous lemmas, we now show that for any prefix of a path with regression bug that can be aligned with a path in the graph constructed by the difference detector \diffDetectOurs that does not visit bug indicator nodes, there exist an alignable path in the original program.

\begin{lemma}\label{lem:align}
Let \(P=(L,\ell_0,G, \ell_\mathrm{err})\) and \(P'=(L',\ell'_0,G', \ell'_\mathrm{err})\) be arbitrary deterministic programs, \diffDetectOurs\((P,P')=(N,E,n_0,\bad)\), and \(\modified\) in the state as at the end of \diffDetectOurs\((P,P')\). 
For all \(\pi' = (\ell'_0,c'_0) \stackrel{g'_1}{\rightarrow}(\ell'_1, c'_1) \stackrel{g'_2}{\rightarrow} \ldots \stackrel{g'_n}{\rightarrow} (\ell'_n, c'_n)\in paths^\mathrm{rb}(P,P')\) and \(0{\leq}k{\leq}n\) if \(n_0\stackrel{g'_1}{\rightarrow}n_1 \stackrel{g'_2}{\rightarrow} \ldots \stackrel{g'_k}{\rightarrow} n_k\in \textnormal{\diffDetectOurs}(P,P')\) and \(\forall 0\leq i\leq k: n_i\notin\bad\), then \(\exists (\ell_0,c_0) \stackrel{g_1}{\rightarrow}(\ell_1, c_1) \stackrel{g_2}{\rightarrow} \ldots \stackrel{g_m}{\rightarrow} (\ell_m, c_m)\in paths(P): c_0=c'_0\wedge\exists f:\{0,\dots,k\}\rightarrow\{0,\dots,m\}:~f\) total, monotonously increasing \(\wedge f(k)=m\wedge\forall 0{\leq}j{\leq}k: n_j=(\ell_{f(j)},\ell'_j)\wedge c'_j=_{|_{\modified(n_j)^-}}c_{f(j)}\).
\end{lemma}

\begin{proof}
Let \(P=(L,\ell_0,G, \ell_\mathrm{err})\) and \(P'=(L',\ell'_0,G', \ell'_\mathrm{err})\) be arbitrary deterministic programs, \diffDetectOurs\((P,P')=(N,E,n_0,\bad)\), and \(\modified\) in the state as at the end of \diffDetectOurs\((P,P')\).
Let \(\pi' = (\ell'_0,c'_0) \stackrel{g'_1}{\rightarrow}(\ell'_1, c'_1) \stackrel{g'_2}{\rightarrow} \ldots \stackrel{g'_n}{\rightarrow} (\ell'_n, c'_n)\in paths^\mathrm{rb}(P,P')\), \(k\in\mathbb{N}\) s.t.\ \(0{\leq}k{\leq}n\), and \(n_0\stackrel{g'_1}{\rightarrow}n_1 \stackrel{g'_2}{\rightarrow} \ldots \stackrel{g'_k}{\rightarrow} n_k\in \textnormal{\diffDetectOurs}(P,P')\) s.t.\ \(\forall 0\leq j\leq k: n_j\notin\bad\) be arbitrary.
Prove by induction, for all \(0\leq i\leq k\) there exists \((\ell_0,c_0) \stackrel{g_1}{\rightarrow}(\ell_1, c_1) \stackrel{g_2}{\rightarrow} \ldots \stackrel{g_{m_i}}{\rightarrow} (\ell_{m_i}, c_{m_i})\in paths(P)\) with \(c_0=c'_0\), and there exists total monotonously increasing function~\(f:\{0,\dots,i\}\rightarrow\{0,\dots,m_i\}\) with \(f(i)=m_i\) and \(\forall 0{\leq}j{\leq}i: n_j=(\ell_{f(j)},\ell'_j)\wedge c'_j=_{|_{{\modified(n_j)}^-}}c_{f(j)}\).

\begin{description}
  \item[Base case ($i=0$)] Define \(c'_0=c_0\) and \(f:=\{(0,0)\}\). By definition, we get \((\ell_0, c_0)\in paths(P)\). Obviously, \(f(0)=0\). Furthermore, \(f\) is total, monotonously increasing. Since \(n_0\notin\bad\), Alg.~\ref{alg:diffcompprop} sets \(n_0=(\ell_0,\ell'_0)\). Furthermore, \(c_0=c'_0\) implies \(c'_0=_{|_{\modified(n_0)^-}}c_{f(0)}\).
  \item[Step case ($i-1\rightarrow i, i\leq k$)]
		Since \(n_0\stackrel{g'_1}{\rightarrow}n_1 \stackrel{g'_2}{\rightarrow} \ldots \stackrel{g'_{k}}{\rightarrow} n_{k}\in \textnormal{\diffDetectOurs}(P,P')\), by definition also prefix \(n_0\stackrel{g'_1}{\rightarrow}n_1 \stackrel{g'_2}{\rightarrow} \ldots \stackrel{g'_{i-1}}{\rightarrow} n_{i-1}\in \textnormal{\diffDetectOurs}(P,P')\).
	By induction, there exists \((\ell_0,c_0) \stackrel{g_1}{\rightarrow}(\ell_1, c_1) \stackrel{g_2}{\rightarrow} \ldots \stackrel{g_{m_{i-1}}}{\rightarrow} (\ell_{m_{i-1}}, c_{m_{i-1}})\in paths(P)\) with \(c_0=c'_0\), and there exists total,  monotonously increasing function~\(f:\{0,\dots,i-1\}\rightarrow\{0,\dots,m_{i-1}\}\) with \(f(i-1)=m_{i-1}\) and \(\forall 0{\leq}j{\leq}i-1: n_j=(\ell_{f(j)},\ell'_j)\wedge c'_j=_{|_{\modified(n_j)^-}}c_{f(j)}\).
Due to Lemma~\ref{lem:extend} and \(i-1<k\leq n\), there exists \((n_{i-1}, g'_{i},n_s)\in E\) such that \(n_s\in\Delta\) or there exists \((\ell_0,c_0) \stackrel{g_1}{\rightarrow}(\ell_1, c_1) \stackrel{g_2}{\rightarrow} \ldots \stackrel{g_{m_{i-1}}}{\rightarrow} (\ell_{m_{i-1}}, c_{m_{i-1}})\ldots \stackrel{g_{{m_{i-1}}+t}}{\rightarrow} (\ell_{{m_{i-1}}+t}, c_{{m_{i-1}}+t})\in paths(P)\) and total, monotonously increasing function~\(f':\{0,\dots,i\}\rightarrow\{0,\dots,{m_{i-1}}+t\}\) with \(f'(i)={m_{i-1}}+t\) and \(\forall 0{\leq}j{<}i: n_j=(\ell_{f'(j)},\ell'_j)\wedge c'_j=_{|_{{\modified(n_j)}^-}}c_{f'(j)}\) 
and \(n_s=(\ell_{f'(i)},\ell'_i)\wedge c'_i=_{|_{{\modified(n_s)}^-}}c_{f'(i)}\).
By definition, \(n_0\stackrel{g'_1}{\rightarrow}n_1 \stackrel{g'_2}{\rightarrow} \ldots \stackrel{g'_{i-1}}{\rightarrow} n_{i-1}\stackrel{g'_{i}}{\rightarrow} n_s\in \textnormal{\diffDetectOurs}(P,P')\).	
Due to Lemma~\ref{lem:deterministic}, \(n_s=n_i\), and, thus, \(n_s\notin\bad\).
Set \(m_{i}=m_{i-1}+t\). 
The hypothesis follows.
\end{description}
\qed
\end{proof}

\subsection{Proof of Theorem~2}

\begin{theorem2}
Algorithm~\ref{alg:diffcompprop} returns a difference graph.
\end{theorem2}

\begin{proof}
Let \(P=(L,\ell_0,G, \ell_\mathrm{err})\) and \(P'=(L',\ell'_0,G', \ell'_\mathrm{err})\) be arbitrary deterministic programs and \diffDetectOurs\((P,P')=(N,E,n_0,\bad)\) the output of Alg.~\ref{alg:diffcompprop}.
By construction, \(n_0\in N\), \(\bad\subseteq N\), and \(E\subseteq N\times G'\times N\).
It remains to be shown that 
for all \(\pi' = (\ell'_0,c'_0) \stackrel{g'_1}{\rightarrow}(\ell'_1, c'_1) \stackrel{g'_2}{\rightarrow} \ldots \stackrel{g'_n}{\rightarrow} (\ell'_n, c'_n)\in paths^\mathrm{rb}(P,P')\) and  \(0{\leq}k{\leq}n\) if \(n_0\stackrel{g'_1}{\rightarrow}n_1 \stackrel{g'_2}{\rightarrow} \ldots \stackrel{g'_k}{\rightarrow} n_k\in \textnormal{\diffDetectOurs}(P,P')\), then there exists \(n_k\stackrel{\cdot}{\rightarrow} \ldots \stackrel{\cdot}{\rightarrow} n_{k+m}\) with \(n_{k+m}\in\Delta\). 

If Alg.~\ref{alg:diffcompprop} returns in line~19, the property is trivially fulfilled.
We continue to show that the property also yields if Alg.~\ref{alg:diffcompprop} does not return in line~19, i.e., it returns in line~21.
Let \( \pi' = (\ell'_0,c'_0) \stackrel{g'_1}{\rightarrow}(\ell'_1, c'_1) \stackrel{g'_2}{\rightarrow} \ldots \stackrel{g'_n}{\rightarrow} (\ell'_n, c'_n)\in paths^\mathrm{rb}(P,P')\), \(k\in\mathbb{N}\) s.t.\ \(0{\leq}k{\leq}n\), and \(n_0\stackrel{g'_1}{\rightarrow}n_1 \stackrel{g'_2}{\rightarrow} \ldots \stackrel{g'_k}{\rightarrow} n_k\in \textnormal{\diffDetectOurs}(P,P')\) be arbitrary.
Furthermore, let \(\modified\) be in the same state as at the end of \diffDetectOurs\((P,P')\). 
If \(n_k\in\bad\), the existence claim follows from path \(n_0\stackrel{g'_1}{\rightarrow}n_1 \stackrel{g'_2}{\rightarrow} \ldots \stackrel{g'_k}{\rightarrow} n_k\).
If \(n_k\notin\bad\), by construction \(\visited\subseteq L\times L'\), \(\bad\subseteq L'\) and by definition  of \(E\) in line~20 no edge leaves a node from~\(\bad\).
Hence, also all \(n_i\) with \(0{\leq}i{\leq}k\) are no elements of \(\bad\).
Due to Lemma~\ref{lem:align}, there exists \((\ell_0,c_0) \stackrel{g_1}{\rightarrow}(\ell_1, c_1) \stackrel{g_2}{\rightarrow} \ldots \stackrel{g_{m_k}}{\rightarrow} (\ell_{m_k}, c_{m_k})\in paths(P)\) s.t.\  \(c_0=c'_0\), and there exists total, monotonously increasing function \(f:\{0,\dots,k\}\rightarrow\{0,\dots,{m_k}\}\) with \(f(k)=m_k\) and \(\forall 0{\leq}i{\leq}k: n_i=(\ell_{f(i)},\ell'_i)\wedge c'_i=_{|_{\modified(n_i)^-}}c_{f(i)}\).
Since \(\pi'\in paths^\mathrm{rb}(P,P')\), we conclude \(\forall 0{\leq}j{\leq}{m_k}: \ell_j\neq \ell_\mathrm{err}\).
Since Alg.~\ref{alg:diffcompprop} does not return in line~19 and \( \ell_0\neq \ell_\mathrm{err}\), also  \(\ell'_0\neq \ell'_\mathrm{err}\).
By construction, \(\ell'_0\neq \ell'_\mathrm{err}\), \(\forall 0{\leq}i{\leq}k: n_i=(\ell_{f(i)},\ell'_i)\), and  \(\forall 0{\leq}j{\leq}{m_k}: \ell_j\neq \ell_\mathrm{err}\), we know that all \(n_i\) with \(0{\leq}i{\leq}k\) are from \((L\setminus\{\ell_\mathrm{err}\})\times(L'\setminus\{\ell'_\mathrm{err}\})\).
Now, let \(i_e\in\mathbb{N}\) be the smallest index such that \(0\leq i_e\leq n\) and \(\ell'_{i_e}=\ell'_\mathrm{err}\).
Since \(\pi'\in paths^\mathrm{rb}(P,P')\) such an \(i_e\) exists.
Furthermore, due to \(n_i=(\cdot,\ell'_i)\in (L\setminus\{\ell_\mathrm{err}\})\times(L'\setminus\{\ell'_\mathrm{err}\})\) for all \(0\leq i\leq k\), also \(i_e>k\).
Show by induction that for all \(t\in\mathbb{N}\) either  \(k+t> i_e\) or \(k+t\leq i_e\) and there exist \(0\leq t_s\leq t\) and \(n_0\stackrel{g'_1}{\rightarrow}n_1 \stackrel{g'_2}{\rightarrow} \ldots \stackrel{g'_k}{\rightarrow} n_k\stackrel{g'_{k+1}}{\rightarrow} \ldots \stackrel{g'_{k+{t_s}}}{\rightarrow} n_{k+{t_s}}\in \textnormal{\diffDetectOurs}(P,P')\) such that  \(n_{k+{t_s}}\in\Delta\) or \(t_s=t\) and \(\forall 0\leq j\leq t_s: n_{k+j}\notin\bad\).
\begin{description}
  \item[Base case ($t=0$)] Follows from \(i_e>k\), \(n_0\stackrel{g'_1}{\rightarrow}n_1 \stackrel{g'_2}{\rightarrow} \ldots \stackrel{g'_k}{\rightarrow} n_k\in \textnormal{\diffDetectOurs}(P,P')\), and \(n_k\notin\bad\). 
  \item[Step case ($t-1\rightarrow t$)]
	If \(k+t> i_e\) or \(k+t\leq i_e\) and there exist \(0\leq t_s\leq t-1\) and \(n_0\stackrel{g'_1}{\rightarrow}n_1 \stackrel{g'_2}{\rightarrow} \ldots \stackrel{g'_k}{\rightarrow} n_k\stackrel{g'_{k+1}}{\rightarrow} \ldots \stackrel{g'_{k+{t_s}}}{\rightarrow} n_{k+{t_s}}\in \textnormal{\diffDetectOurs}(P,P')\) such that \(n_{k+{t_s}}\in\Delta\), nothing needs to be shown.
	Assume \(k+t\leq i_e\) and there do not exist \(0\leq t_s\leq t-1\) and \(n_0\stackrel{g'_1}{\rightarrow}n_1 \stackrel{g'_2}{\rightarrow} \ldots \stackrel{g'_k}{\rightarrow} n_k\stackrel{g'_{k+1}}{\rightarrow} \ldots \stackrel{g'_{k+{t_s}}}{\rightarrow} n_{k+{t_s}}\in \textnormal{\diffDetectOurs}(P,P')\) such that \(n_{k+{t_s}}\in\Delta\).
	By induction hypothesis, there exists \(n_0\stackrel{g'_1}{\rightarrow}n_1 \stackrel{g'_2}{\rightarrow} \ldots \stackrel{g'_k}{\rightarrow} n_k\stackrel{g'_{k+1}}{\rightarrow} \ldots \stackrel{g'_{k+{t-1}}}{\rightarrow} n_{k+{t-1}}\in \textnormal{\diffDetectOurs}(P,P')\) such that \(\forall 0\leq j\leq t-1: n_{k+j}\notin\bad\).
	Since \(k+t\leq i_e\leq n\), also \(k+t-1<n\).
	Furthermore, from before we know that \(\forall 0\leq j\leq k: n_{j}\notin\bad\).
	Hence, \(\forall 0\leq j\leq k+t-1: n_j\notin\bad\).
	Combining Lemma~\ref{lem:align} and Lemma~\ref{lem:extend}, we conclude there exists \((n_{k+t-1}, g'_{k+t},n_{k+t})\in E\) with \(n_{k+t}\in\Delta\vee n_{k+1}=(\cdot,\ell'_{k+t})\).
	By construction of \(\Delta\), if \(n_{k+1}=(\cdot,\ell'_{k+t})\), then \(n_{k+1}\notin\Delta\).
	By definition, \(n_0\stackrel{g'_1}{\rightarrow}n_1 \stackrel{g'_2}{\rightarrow} \ldots \stackrel{g'_k}{\rightarrow} n_k\stackrel{g'_{k+1}}{\rightarrow} \ldots \stackrel{g'_{k+{t-1}}}{\rightarrow} n_{k+{t-1}}\stackrel{g'_{k+{t}}}{\rightarrow} n_{k+{t}}\in \textnormal{\diffDetectOurs}(P,P'): n_{k+{t}}\in\Delta\vee\forall 0\leq j\leq t: n_{k+j}\notin\Delta\).
	The induction hypothesis follows.
\end{description}
Consider \(t=i_e-k\).
Since \(i_e>k\), we know \(t>0\).
Due to induction, there exist \(0\leq t_s\leq i_e-k\) and \(n_0\stackrel{g'_1}{\rightarrow}n_1 \stackrel{g'_2}{\rightarrow} \ldots \stackrel{g'_k}{\rightarrow} n_k\stackrel{g'_{k+1}}{\rightarrow} \ldots \stackrel{g'_{k+{t_s}}}{\rightarrow} n_{k+{t_s}}\in \textnormal{\diffDetectOurs}(P,P')\) such that \(n_{k+t_s}\in\Delta\) or \(t_s=i_e-k\) and \(\forall 0\leq j\leq t_s: n_{k+j}\notin\bad\).
Show by contradiction that there do not exist \(0\leq t_s\leq i_e-k\) and \(n_0\stackrel{g'_1}{\rightarrow}n_1 \stackrel{g'_2}{\rightarrow} \ldots \stackrel{g'_k}{\rightarrow} n_k\stackrel{g'_{k+1}}{\rightarrow} \ldots \stackrel{g'_{k+{t_s}}}{\rightarrow} n_{k+{t_s}}\in \textnormal{\diffDetectOurs}(P,P')\) such that \(t_s=i_e-k\) and \(\forall 0\leq j\leq t_s: n_{k+j}\notin\bad\).
Assume the contrary.
From before we know that \(\forall 0\leq j\leq k: n_{j}\notin\bad\).
Hence,  there exists \(n_0\stackrel{g'_1}{\rightarrow}n_1 \stackrel{g'_2}{\rightarrow} \ldots \stackrel{g'_k}{\rightarrow} n_k\stackrel{g'_{k+1}}{\rightarrow} \ldots \stackrel{g'_{i_e}}{\rightarrow} n_{i_e}\in \textnormal{\diffDetectOurs}(P,P')\) s.t.\ \(\forall 0\leq j\leq i_e: n_{j}\notin\bad\).
Due to Lemma~\ref{lem:align}, there exists \(\pi=(\ell_0,c_0) \stackrel{g_1}{\rightarrow}(\ell_1, c_1) \stackrel{g_2}{\rightarrow} \ldots \stackrel{g_m}{\rightarrow} (\ell_m, c_m)\in paths(P)\) with \(c_0=c'_0\) and \(n_{i_e}=(\ell_m, \ell'_{i_e})\).
Since \(\pi'\in paths^\mathrm{rb}(P,P')\), we know \(\pi\notin paths^\mathrm{err}\).
Hence, \(\ell_m\neq\ell_\mathrm{err}\) and since \(\ell'_{i_e}=\ell'_\mathrm{err}\), we conclude \((\ell_m, \ell'_{i_e})\in (L\setminus\{\ell_\mathrm{err}\})\times\{\ell'_\mathrm{err}\}\).
Since \(\ell'_0\neq \ell'_\mathrm{err}\) and by construction \(n_{i_e}=(\ell_m,\ell'_\mathrm{err})\notin\Delta\) implies \((\ell_m,\ell'_\mathrm{err})\in\visited\), we conclude \((\ell_m,\ell'_\mathrm{err})\) must be added in line~32.
To be added in line~32, \((\ell_m,\ell'_\mathrm{err})\notin (L\setminus\{\ell_\mathrm{err}\})\times\{\ell'_\mathrm{err}\}\), which contradicts \((\ell_m, \ell'_{i_e})\in (L\setminus\{\ell_\mathrm{err}\})\times\{\ell'_\mathrm{err}\}\).
Hence,  there do not exist \(0\leq t_s\leq i_e-k\) and \(n_0\stackrel{g'_1}{\rightarrow}n_1 \stackrel{g'_2}{\rightarrow} \ldots \stackrel{g'_k}{\rightarrow} n_k\stackrel{g'_{k+1}}{\rightarrow} \ldots \stackrel{g'_{k+{t_s}}}{\rightarrow} n_{k+{t_s}}\in \textnormal{\diffDetectOurs}(P,P')\) with \(t_s=i_e-k\) and \(\forall 0\leq j\leq t_s: n_{k+j}\notin\bad\).
Thus, there exist \(0\leq t_s\leq i_e-k\) and \(n_0\stackrel{g'_1}{\rightarrow}n_1 \stackrel{g'_2}{\rightarrow} \ldots \stackrel{g'_k}{\rightarrow} n_k\stackrel{g'_{k+1}}{\rightarrow} \ldots \stackrel{g'_{k+{t_s}}}{\rightarrow} n_{k+{t_s}}\in \textnormal{\diffDetectOurs}(P,P')\) such that \(n_{k+t_s}\in\Delta\).
\qed
\end{proof}

\end{document}